\documentclass[11pt,letterpaper]{article}
\usepackage[T1]{fontenc}
\usepackage[utf8]{inputenc}
\usepackage[margin=1in]{geometry}
\usepackage{amsmath,amssymb,amsthm}
\usepackage{newtxtext,newtxmath}

\usepackage{mathtools,bm}
\usepackage{graphicx,array,booktabs}
\usepackage{flafter}
\usepackage[font=small,labelfont=bf]{caption}
\usepackage[authoryear,round]{natbib}
\bibpunct[, ]{(}{)}{,}{a}{}{,}
\usepackage[hidelinks]{hyperref}
\hypersetup{
  pdftitle={Stability and Optimal Design of Adaptive Linear Bus Holding on Open Chains and Loops},
  pdfauthor={Li Zhen, Weihua Gu, Minyu Shen, Wenbo Fan, Feng Xiao},
  pdfsubject={Bus holding control and transit service reliability},
  pdfkeywords={bus bunching, optimal holding control, schedule and headway reliability, finite fleet, boundary conditions, Fourier analysis}
}
\theoremstyle{plain}
\newtheorem{proposition}{Proposition}
\newtheorem{corollary}{Corollary}
\theoremstyle{remark}
\newtheorem{remark}{Remark}
\DeclareMathOperator{\Var}{Var}
\DeclareMathOperator{\tr}{tr}
\DeclareMathOperator{\diag}{diag}

\title{Stability and Optimal Design of Adaptive Linear Bus Holding on Open Chains and Loops}
\author{%
\normalsize Li Zhen\textsuperscript{1}\quad
Weihua Gu\textsuperscript{1}\thanks{Corresponding author: Weihua Gu. Email: \href{mailto:weihua.gu@polyu.edu.hk}{weihua.gu@polyu.edu.hk}.}\quad
Minyu Shen\textsuperscript{2}\quad
Wenbo Fan\textsuperscript{1}\quad
Feng Xiao\textsuperscript{3}\\[0.8em]
\parbox{\dimexpr\textwidth-2\tabcolsep\relax}{\centering\small
\textsuperscript{1}Department of Electrical and Electronic Engineering,\\
The Hong Kong Polytechnic University, Hong Kong, China\\[0.4em]
\textsuperscript{2}School of Management Science and Engineering,\\
Southwestern University of Finance and Economics, Chengdu, China\\[0.4em]
\textsuperscript{3}Business School, Sichuan University, Chengdu, China}}
\date{}

\begin{document}
\maketitle
\begin{abstract}
Bus holding can suppress bunching, but its stability and performance depend on the feedback structure and on how the bus sequence is bounded. This paper develops a common analytical framework for adaptive linear holding on finite open chains and loops and relates both to the bi-infinite system studied in the literature. Matrix and Fourier methods yield necessary and sufficient stability conditions for current- and leading-bus kernels on open chains and general cyclic kernels on loops, together with computable covariance bounds on open chains and exact variance formulas on loops. The analysis shows when periodic closure changes stability, when coefficient-sum stability tests in the literature are conservative, and why a stable open chain can still accumulate substantial variance before boundary effects arrive. For loops with same-lap feedback, modal analysis links control coefficients to fleet-wide schedule drift and to relative headway patterns, explaining the behavior of existing rules and yielding closed-form Pareto-optimal designs for schedule--holding and headway--holding tradeoffs. Cross-lap output corrections relate these designs to physical headway and holding performance. Numerical experiments show that the analytical design and a three-tap approximation, both obtained from closed-form coefficients and a one-dimensional search, outperform four established controls at the tested moderate and loose headway-variance budgets, with holding variances within 0.35\% and 3.5\% of the full cyclic numerical frontier across these budgets and other parameter variations. The framework thus provides a tractable design procedure and an explanation of how feedback structure and boundary conditions shape service reliability and holding effort.
\end{abstract}

\noindent\textbf{Keywords:} bus bunching; optimal holding control; schedule and headway reliability; finite fleet; boundary conditions; Fourier analysis

\section{Introduction}\label{sec:introduction}

Bus bunching is common in high-frequency bus services. A bus that falls behind picks up more passengers, dwells longer, and falls further behind. Its follower picks up fewer passengers and catches up. This process increases passenger wait, produces uneven loads, and blocks following buses' access to berths \citep{newell1964,shen2023berth}. Bus holding is a widely studied way to control bunching. At a control point, a bus can be held when it is early, too close to the bus ahead, or too far from the bus behind. Holding can improve schedule and headway reliability, but it also delays onboard passengers and often requires planned slack in its schedule \citep{barnett1974,daganzo2009,xuan2011}. A holding policy must therefore balance service reliability against holding cost.

Bus operations form a dynamical system in which disturbances \citep[e.g.,][]{gu2014signals} propagate between neighboring buses and onward through the fleet. Adaptive holding adds feedback by adjusting holding times to current operating conditions. System stability and performance depend jointly on the holding rule and whether the bus sequence ends or wraps around. Connecting the last bus back to the first can give disturbances an additional path to travel among buses, potentially changing whether a holding rule suppresses or amplifies them.

Two common operating settings, open chains and loops, illustrate this distinction. On an open corridor, each trip runs from one terminal to another, with terminal layover time commonly used to absorb running-time variation \citep{tflguidelines}. Figure~\ref{fig:three-systems}(a) represents this as a finite open-chain corridor of terminal-dispatched trips, with a first bus, a last bus, and no wrap-around. On a loop corridor, the same physical fleet repeatedly circulates through a closed route, as in urban circulator and shuttle services \citep{ftacirculator}. Figure~\ref{fig:three-systems}(b) represents this as a finite periodic fleet whose last and first buses are adjacent, so disturbances can pass from the last bus back to the first. In addition to these two operating settings, the literature considers a theoretical idealization: the bi-infinite system \citep{xuan2011}, whose bus sequence extends indefinitely in both directions, as illustrated in Figure~\ref{fig:three-systems}(c).

\begin{figure}[htbp]
\centering
\includegraphics[width=\textwidth]{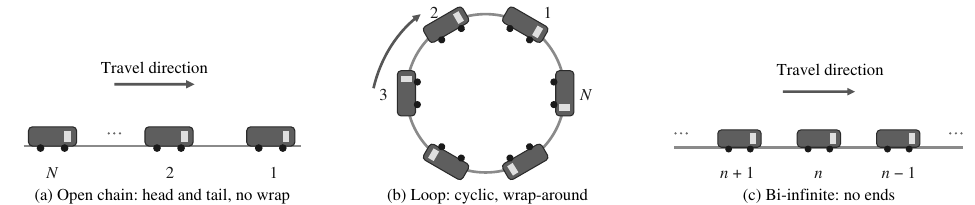}
\caption{Three boundary settings for bus sequences.\label{fig:three-systems}}

\end{figure}

These settings impose different boundary conditions on interactions between buses: interactions stop at the ends of an open chain, connect the last bus to the first on a loop, and extend indefinitely in the bi-infinite idealization. Under linear holding rules which adjust holding times using weighted sums of schedule or headway deviations \citep{xuan2011}, the same weights (i.e., control coefficients) link buses differently across these settings. As shown below, coefficients that stabilize an open chain can make a loop unstable, or vice versa. Control design and performance assessment must therefore account for how the bus sequence is bounded.

The bus-control literature uses all three settings, but the boundary condition is usually treated as part of the chosen model rather than as a factor to be compared. \citet{daganzo2009} studies a one-sided infinite bus sequence, \citet{xuan2011} analyze a bi-infinite system, and other studies examine finite periodic fleets on loops \citep{daganzo2011,bartholdi2012,liang2016,zhang2018}. \citet{daganzo2011} suggest that open-route results can approximate those of a sufficiently long loop when disturbances dissipate over much shorter distances. For the general linear holding family, however, existing studies do not provide a common framework relating boundary conditions and control coefficients to stability and performance. Analytical optimization of the reliability--holding tradeoff also remains largely confined to special cases.

This paper studies homogeneous bus operations under the three boundary settings using this linear holding family, which includes widely studied schedule- and headway-based rules. Here, homogeneous means the scheduled headway, demand intensity, disturbance variance, and control coefficients are fixed across buses and control points. The linear structure of holding rules allows control coefficients to be related directly to stability, performance, and optimal design.

Four questions remain. \textbf{(1)} How do the boundary conditions of finite bus sequences and feedback coefficients jointly determine stability? \textbf{(2)} Can schedule and headway reliability and holding effort be quantified, and their tradeoffs optimized analytically? \textbf{(3)} What explains differences in the performance of holding rules, and how does the bi-infinite idealization differ from and relate to finite-fleet systems (open chain or loop)? \textbf{(4)} To limit the information required for each holding decision, can a rule using only the current bus and its immediate neighbors approach the optimal headway--holding tradeoff?

Our major contributions are:

\begin{enumerate}
\item \textbf{Exact stability conditions.} Matrix analysis and the Discrete Fourier Transform (DFT) yield necessary and sufficient stability conditions for general leader-only kernels (using deviations from the current bus and buses ahead) on finite open chains and general cyclic kernels on finite loops. To our knowledge, these are the first such characterizations of each full kernel family. On leader-only open chains, stability depends only on the current-bus coefficient, the sole kernel coefficient retained by Simple Control (SC). On loops, the full kernel and fleet size enter the stability condition. The same kernel can therefore be stable under one boundary and unstable under another.

\item \textbf{Analytical performance and optimal design.} On loops, Fourier decomposition yields exact performance formulas and closed-form Pareto designs over the full cyclic coefficient family (i.e., for linear holding rules that can use deviations from all buses), covering both schedule--holding and headway--holding tradeoffs. For open chains, we derive steady-state covariance bounds for general leader-only kernels.

\item \textbf{Mechanisms behind control performance.} Modal analysis links control coefficients to fleet-wide schedule drift and alternating short--long headway patterns. It explains why headway-based rules can maintain regular spacing while schedule deviations grow. For a fixed kernel, the finite-loop stability criterion approaches the bi-infinite criterion as fleet size grows. The same analysis shows how cancellation between coefficients can make the sufficient stability condition in \citet{xuan2011} conservative.

\item \textbf{Near-optimal local control.} Numerical comparisons of headway and holding variances show that a three-tap approximation to the analytical kernel, using the current bus and its immediate neighbors, closely approaches the numerically optimized full cyclic frontier and outperforms the named control laws at matched headway reliability in the tested moderate and loose budget ranges. {A separate same-lap sensitivity analysis accounts for prescribed independent additive errors in future bus information. Retuning a single design parameter preserves savings over SC in the moderate-error cases studied, whereas larger errors combined with tight headway-variance budgets can eliminate that advantage.}
\end{enumerate}

Section~\ref{sec:literature} reviews and classifies relevant bus-control studies by information structure and boundary conditions of the bus sequence. Section~\ref{sec:modelclasses} presents the common local bus-motion model and the three boundary conditions. Sections~\ref{sec:openchain} and~\ref{sec:loop} analyze the finite open-chain and loop systems, respectively. Section~\ref{sec:finite-boundary-discussion} discusses finite-horizon boundary effects. Section~\ref{sec:numerics} gives numerical case studies, and Section~\ref{sec:conclusions} concludes.

\section{Literature Review}\label{sec:literature}

Bus-bunching controls include holding, speed regulation \citep{bian2023}, conditional signal priority \citep{anderson2020}, stop-skipping \citep{gkiotsalitis2021}, boarding limits \citep{delgado2012}, and vehicle substitution \citep{petit2019}. In practice, holding can also be realized by reducing the bus's cruising speed. Holding has been studied using optimization-based, simulation-based, empirical, and analytical approaches. Optimization-based studies formulate holding decisions under passenger demand, vehicle capacity, and network constraints \citep{eberlein2001,berrebi2015,hernandez2015,dai2019}, while field studies evaluate control performance and implementability under realistic operating conditions \citep{cats2014,berrebi2018b}.

We focus on the analytical literature because our central questions---disturbance propagation, stability, and covariances---depend directly on the structure of the controlled dynamics \citep[e.g.,][]{daganzo2009,xuan2011,bartholdi2012}. Consistent with our analysis, we organize the review around holding control laws, boundary conditions of the bus sequence, and stability and performance. In addition to SC, we use the following abbreviations: forward-headway control (FHC), backward-headway control (BHC), and two-way-looking headway control (TWHC). Table~\ref{tab:literature} summarizes representative analytical studies.

\begin{table}[t]
\caption{Structural classification of representative analytical bus-control studies.\label{tab:literature}}
\centering
\footnotesize\setlength{\tabcolsep}{2pt}

\begin{tabular}{
>{\raggedright\arraybackslash}p{0.140\linewidth}
>{\raggedright\arraybackslash}p{0.130\linewidth}
>{\raggedright\arraybackslash}p{0.160\linewidth}
>{\raggedright\arraybackslash}p{0.200\linewidth}
>{\raggedright\arraybackslash}p{0.320\linewidth}}
\toprule
Study
& Control law
& Bus sequence
& Boundary treatment
& Stability/performance \\
\midrule

\citet{daganzo2009}
& FHC
& One-sided infinite open chain
& Fixed head boundary
& Uniform headway-variance boundedness \\
\addlinespace[4pt]

\citet{daganzo2011}
& Two-way-looking speed control
& Finite periodic loop
& Loop boundary; fixed cyclic bus ordering
& Spacing-deviation boundedness and variance \\
\addlinespace[4pt]

\citet{xuan2011}
& General linear control; SC
& Bi-infinite system
& No head or tail boundary
& Variance bounds; SC performance and design \\
\addlinespace[4pt]

\citet{bartholdi2012}
& BHC
& Finite periodic loop
& Loop boundary; cyclic re-indexing at a control point
& Deterministic common-headway convergence \\
\addlinespace[4pt]

\citet{argote2015}
& Multiline SC
& One-sided sequences on interacting lines
& No explicit head-boundary condition
& Schedule-deviation variance boundedness \\
\addlinespace[4pt]

\citet{liang2016}
& TWHC
& Finite periodic loop
& Loop boundary; cyclic re-indexing at a control point
& Numerical headway convergence and bunching boundaries \\
\addlinespace[4pt]

\citet{andres2017}
& TWHC with prediction
& Finite open chain
& Fixed head and tail boundaries
& Finite-route worst-case headway-deviation bounds \\
\addlinespace[4pt]

\citet{zhang2018}
& TWHC
& Finite periodic loop
& Loop boundary; cyclic re-indexing at a control point
& Deterministic common-headway convergence; stochastic headway-variance convergence \\
\bottomrule
\end{tabular}

\end{table}

\subsection{Holding Control Laws}\label{subsec:control-laws}

Analytical holding laws differ in both feedback coefficients and information requirements. FHC uses the observed headway from the leading bus and directly counteracts the demand-driven amplification of headway deviations \citep{daganzo2009}. BHC and TWHC incorporate following-bus information, as in self-equalizing holding schemes on loops \citep{bartholdi2012,liang2016,zhang2018}, where feedback brings headways toward a common value without a prescribed target headway. \citet{daganzo2011} study a related two-way-looking speed-control formulation. Note that when holding is decided, the forward headway has been observed, whereas the following bus has not yet reached the same control point. Rules using the backward headway (e.g., BHC and TWHC) therefore require an estimate or prediction.

\citet{xuan2011} develop the general linear holding framework adopted here, encompassing linear forms of FHC, BHC, and TWHC. Each bus's updated schedule deviation is a weighted sum of its own and other buses' deviations; we call the ordered weights the \emph{control kernel}. They also introduce SC, derive its performance formulas and optimal design formulation, and show numerically that it is nearly optimal in schedule-deviation--slack tradeoffs and outperforms headway-based benchmarks in headway--slack tradeoffs.

A control kernel specifies which relative bus positions enter the holding decision and with what weights, but it does not specify how the bus sequence is bounded. For example, a kernel may use each bus's own deviation and that of the bus one position earlier in the sequence. For the first bus, that position lies outside an open-chain fleet but refers to the last bus in the loop model. The same coefficients thus define different feedback connections under the boundary conditions discussed next.

\subsection{Boundaries of the Bus Sequence}\label{subsec:boundaries}

Boundary conditions determine how interactions along the ordered bus sequence are terminated at the ends, connected into a loop, or extended indefinitely, rather than simply whether the physical route is open or circular. Open-chain models describe successive dispatched trips, loop models retain the repeated circulation of a finite fleet, and the bi-infinite idealization describes limiting behavior away from sequence boundaries.

For open chains, \citet{daganzo2009} analyzes FHC on a one-sided infinite bus sequence with a fixed head boundary and on-time dispatching. \citet{argote2015} extend SC to one-sided bus-trip sequences on interacting lines, preserving decoupled dynamics in the linear regime. In contrast, \citet{andres2017} study a finite bus sequence with fixed headways at both ends.

For finite periodic loops, \citet{daganzo2011} note that a sufficiently long loop can approximate an open system if disturbances dissipate over much shorter distances. The models of \citet{bartholdi2012}, \citet{liang2016}, and \citet{zhang2018} cyclically relabel buses at a control point, assigning label~1 to each newly arriving bus. The loop boundary connects the sequence's end to its beginning, removing distinct head and tail boundaries.

\citet{xuan2011} represent bus dynamics through convolution over a bi-infinite sequence, taking limits away from the initial bus and station boundaries. By suppressing finite-boundary effects, this idealization makes analytical stability and variance results tractable for general linear kernels.

Existing results remain tied to particular combinations of control laws, boundary settings, and operating assumptions; the resulting gaps in stability and performance are examined next.

\subsection{Stability and Performance}\label{subsec:stability-performance}

Bus-holding studies assess passenger impacts, service reliability, and control effort. Passenger impacts include waiting time, in-vehicle delay, and denied boarding \citep{gkiotsalitis2021}. Service reliability is commonly measured by schedule and headway deviations and their variability \citep{xuan2011}. Improved headway regularity is often associated with lower passenger waiting. Control effort measures include holding time, schedule slack, and commercial-speed loss \citep{zhao2006,berrebi2018a}. Stability analyses focus on convergence or boundedness of deviations, including bounded variances under disturbances. Characterizing reliability--holding tradeoffs further requires quantifying the reliability achieved and the associated holding effort.

In open-chain settings, analytical results include uniform headway-variance bounds for FHC \citep{daganzo2009}, schedule-deviation variance bounds for SC \citep{argote2015}, and finite-route worst-case headway-deviation bounds \citep{andres2017}. For self-equalizing loop controls, \citet{bartholdi2012} prove deterministic convergence to equal headways under common-speed dynamics. \citet{liang2016} primarily examine convergence rates and bunching boundaries numerically. \citet{zhang2018} analyze deterministic and stochastic headway regulation under their two-way-looking self-equalizing rule, with headway-dependent dwell-time effects examined numerically. Among these studies, the bi-infinite analysis of \citet{xuan2011} accommodates the broadest class of linear holding laws. They derive a sufficient condition for bounded schedule-deviation variance under independent, identically distributed disturbances and prove that this variance diverges for every linear headway-based control using deviations from only finitely many buses, including FHC, BHC, and TWHC. As they emphasize, bounded headway variance does not imply bounded schedule-deviation variance: regular spacing can coexist with growing schedule deviations. Despite these advances, existing studies lack a unified framework providing necessary and sufficient stability conditions for general linear control under different boundary settings.

Beyond stability, control design must account for the holding effort required to achieve a given level of reliability. \citet{xuan2011} formulate minimum-slack problems under schedule- or headway-reliability constraints and solve them numerically for kernels with few nonzero coefficients, with closed-form results for SC. \citet{berrebi2018a}, \citet{gu2025}, and \citet{shen2026} compare headway--holding tradeoffs by simulation. The Pareto frontier for general linear holding laws still lacks an analytical characterization.

\section{Common Model and Three Boundary Conditions}\label{sec:modelclasses}

\subsection{Local Bus-Motion Model and Adaptive Linear Holding Law}\label{subsec:local-model}

Consider a homogeneous bus system controlled at a sequence of control points, without additional delays from buses queuing for access to stops \citep{gu2015queues,shen2019capacity}. For bus $n$ at control point $s$, let $t_{n,s}$ and $a_{n,s}$ denote scheduled and actual arrival times, with schedule deviation $\varepsilon_{n,s}=a_{n,s}-t_{n,s}$. Throughout the paper, \(t_{n,s}\) is a \emph{virtual schedule}: the internal reference trajectory for control and stability analysis, not necessarily the passenger-facing timetable. If a published timetable is needed at every stop, it can be obtained by subtracting a deterministic buffer from the virtual schedule to minimize the risk that passengers arriving on time miss an early bus. Unless stated otherwise, all schedule deviations are relative to the virtual schedule.

Let $H$ be the scheduled headway, $d_s$ the slack built into the virtual schedule at control point $s$, $D_{n,s}$ the holding time, $C_s$ the average cruising time from control point $s$ to $s+1$, and $v_{n,s+1}$ the zero-mean cruising-time disturbance from traffic interactions. The homogeneous dimensionless demand rate, $\beta$, is the product of passenger arrival rate and mean boarding time per passenger.

The scheduled and actual motion equations, and the schedule deviation propagation are
\begin{equation}
\begin{aligned}
t_{n,s+1}&=t_{n,s}+\beta H+d_s+C_s,\\
a_{n,s+1}&=a_{n,s}+\beta h_{n,s}+D_{n,s}+C_s+v_{n,s+1},\\
\varepsilon_{n,s+1}&=\varepsilon_{n,s}+\beta(\varepsilon_{n,s}-\varepsilon_{n-1,s})+(D_{n,s}-d_s)+v_{n,s+1},
\end{aligned}
\label{eq:deviation-propagation}
\end{equation}
where $h_{n,s}=a_{n,s}-a_{n-1,s}=H+\varepsilon_{n,s}-\varepsilon_{n-1,s}$ is the headway to the preceding bus. We consider the following adaptive linear holding law:
\begin{equation}
D_{n,s}=d_s-(1+\beta)\varepsilon_{n,s}+\beta\varepsilon_{n-1,s}+\sum_{p=0}^{m}f_p\varepsilon_{n-p,s}.
\label{eq:general-law}
\end{equation}
The coefficients $f_p$ form the control kernel. We initially use $p=0,\ldots,m$, corresponding to current- and leading-bus inputs on an open chain; negative offsets allow following-bus inputs.

Substituting \eqref{eq:general-law} into \eqref{eq:deviation-propagation} yields the local dynamics under this holding law
\begin{equation}
\varepsilon_{n,s+1}=\sum_{p=0}^{m}f_p\varepsilon_{n-p,s}+v_{n,s+1}.
\label{eq:local-closed}
\end{equation}
For a given kernel, the systems differ only in how the bus index $n-p$ is interpreted at the boundary. Unless stated otherwise, any required but unobserved deviation is assumed to be predicted exactly at the holding decision. Perfect-information benchmarks have also been used in analytical bus-control studies \citep{xuan2011,berrebi2015}.

We use the named rules defined in \citet{xuan2011} as benchmarks: the triples $(f_{-1},f_0,f_1)$ are $(0,f_0,0)$ for SC, $(0,1-\alpha,\alpha)$ for FHC, $(\alpha,1+\beta-\alpha,-\beta)$ for BHC, and $(\alpha,1-2\alpha,\alpha)$ for TWHC; all other coefficients are zero. Here, $\alpha$ sets the strength of headway feedback. Consistent with the literature \citep{xuan2011,bartholdi2012}, we use $0<\alpha<1$ for FHC and BHC and $0<\alpha<1/2$ for TWHC.

\subsection{Three System Classes}\label{subsec:three-systems}

The local recursion \eqref{eq:local-closed} yields three systems with different boundary conditions on the bus sequence.

\textit{Finite open-chain linear corridor.} In this abstraction, $N$ counts scheduled bus dispatches from the origin terminal. Trips depart at the prescribed headway $H$, so terminal dispatch acts as a schedule reset. This represents terminal-dispatched service on noncyclic routes, as studied in, e.g., \citet{hickman2001} and \citet{andres2017}. The chain has a head and a tail, with no wrap-around from bus $N$ to bus $1$. Since the first bus has no leader within the modeled fleet, headway performance is averaged over buses $2,\ldots,N$. Section~\ref{sec:openchain} shows that, with current- and leading-bus feedback, stability depends only on the current-bus coefficient $f_0$, through $|f_0|<1$.

\textit{Finite periodic fleet on a loop corridor.} Here $N$ is the physical fleet size, with buses ordered cyclically \citep[e.g.,][]{bartholdi2012,zhang2018}. Section~\ref{sec:loop} shows that the loop boundary brings the full control kernel into the stability condition.

\textit{Bi-infinite system.} Buses are indexed by all integers, with no head, tail, or finite cyclic closure. Under the homogeneous holding law, the translation-invariant control dynamics are represented by convolution along the bus sequence. This is the idealized setting used in the stability analysis of \citet{xuan2011}.

Figure~\ref{fig:three-systems} contrasts the three settings, and Table~\ref{tab:systems} summarizes their main stability characteristics.

\begin{table}[htbp]
\caption{Overview of boundary settings, feedback classes, and stability results.\label{tab:systems}}
\centering
\footnotesize\setlength{\tabcolsep}{3pt}
\begin{tabular}{
>{\raggedright\arraybackslash}p{0.33\textwidth}
>{\raggedright\arraybackslash}p{0.19\textwidth}
>{\raggedright\arraybackslash}p{\dimexpr0.48\textwidth-6\tabcolsep\relax}}
\toprule
System and feedback class & Boundary setting & Stability summary \\
\midrule
Open chain: current- and leading-bus feedback & finite $N$, no wrap-around & Depends only on the current-bus coefficient: $|f_0|<1$. \\
Open chain: two-sided three-tap feedback & finite $N$, no wrap-around & Depends on the three-tap kernel and $N$. \\
Finite periodic fleet / loop corridor & finite $N$, cyclic & Depends on the full kernel and $N$. \\
Bi-infinite system & infinite, no ends & Limit of the loop criterion as $N\to\infty$, with the kernel fixed. \\
\bottomrule
\end{tabular}

\end{table}

We next analyze stability and performance, first on finite open chains and then on loops, relating the results to the bi-infinite system.

\section{Open-Chain Stability and Performance on a Finite-Bus Corridor}\label{sec:openchain}

Sections~\ref{subsec:open-matrix}--\ref{subsec:simple-open-analysis} study current- and leading-bus feedback ($p\ge0$) on a finite open chain. The lower-shift representation yields its stability boundary, transient behavior, and computable covariance bounds; SC then supplies exact performance--cost formulas and a constrained design rule. Section~\ref{subsec:open-two-sided} admits following-bus feedback and derives an explicit stability condition for the special three-tap kernel, delineating the scope of the preceding triangular results.

\subsection{Matrix Representation, Stability, and Transient Effects for Leader-Only Kernels}\label{subsec:open-matrix}

For the finite open-chain system, define the schedule-deviation vector $\boldsymbol{\varepsilon}_{s}=(\varepsilon_{1,s},\ldots,\varepsilon_{N,s})^{\top}$ and the disturbance vector $\mathbf{v}_{s}$ analogously. As specified, we start with $p\ge0$. The control dynamics from \eqref{eq:local-closed} are
\begingroup

\begin{equation}
\boldsymbol{\varepsilon}_{s+1}=\mathbf{A}_{\mathrm{o}}\boldsymbol{\varepsilon}_{s}+\mathbf{v}_{s+1},\qquad
\mathbf{A}_{\mathrm{o}}=\sum_{p=0}^{m}f_p\mathbf{L}^{p},\qquad
\mathbf{L}=\begin{bmatrix}
0&0&0&\cdots&0\\
1&0&0&\cdots&0\\
0&1&0&\cdots&0\\
\vdots&&\ddots&\ddots&\vdots\\
0&0&\cdots&1&0
\end{bmatrix}_{N\times N}.
\label{eq:open-state}
\end{equation}
\endgroup
The lower-shift matrix $\mathbf{L}$ inserts a zero at the head and shifts each leading-bus deviation toward the tail. Since $\mathbf{L}$ is lower triangular, so is $\mathbf{A}_{\mathrm{o}}$. Applying $\mathbf{L}$ $N$ times gives $\mathbf{L}^{N}=\mathbf{0}$, i.e., $\mathbf{L}$ is nilpotent \citep{horn2013}.

Since every diagonal entry of $\mathbf{L},\mathbf{L}^2,\ldots,\mathbf{L}^m$ is zero, the diagonal of the lower-triangular $\mathbf{A}_{\mathrm{o}}$ comes entirely from $f_0\mathbf{I}$. The following result is immediate.

\begin{proposition}\label{prop:open-spectrum}
For the leader-only finite-bus open-chain system \eqref{eq:open-state}, every eigenvalue of $\mathbf{A}_{\mathrm{o}}$ equals $f_0$, so its spectral radius is $\rho(\mathbf{A}_{\mathrm{o}})=|f_0|$; therefore, its zero-input dynamics are asymptotically stable if and only if $|f_0|<1$.
\end{proposition}

Within this one-sided open-chain class, higher-order coefficients ($p>0$) can reshape transients, boundary behavior, covariance distribution, and finite-horizon performance, but do not move the long-run stability threshold. Both introducing opposite-direction feedback on the open chain and connecting the last bus back to the first can make interbus coupling enter the stability condition. The $f_0$-only criterion therefore describes only the triangular control class. Two-sided feedback requires the broader stability analysis as exemplified in Section~\ref{subsec:open-two-sided}.

SC retains precisely the coefficient governing asymptotic stability in this class, with transition matrix $\mathbf A_{\mathrm o}=f_0\mathbf I$.

The power $\mathbf{A}_{\mathrm{o}}^{s}$ maps initial deviations to their values after $s$ control-point transitions without further disturbances. To see how the remaining coefficients enter finite-horizon behavior, write $\mathbf{A}_{\mathrm{o}}=f_0\mathbf{I}+\mathbf{R}$ with $\mathbf{R}:=\sum_{p=1}^{m}f_p\mathbf{L}^{p}$, a polynomial in $\mathbf{L}$ with no identity term and hence also nilpotent. Let $\nu$ be the smallest positive integer with $\mathbf{R}^{\nu}=\mathbf{0}$, so $\nu\le N$. Since $\mathbf{I}$ and $\mathbf{R}$ commute,
\begin{equation}
\mathbf{A}_{\mathrm{o}}^{s}=\sum_{k=0}^{\min\{s,\nu-1\}}\binom{s}{k}f_0^{s-k}\mathbf{R}^{k}.
\label{eq:open-binomial}
\end{equation}
Equation~\eqref{eq:open-binomial} shows how higher-order coefficients affect finite-horizon behavior without moving the stability boundary. The exponential growth or decay rate is governed by $f_0^s$; the remaining coefficients enter through nilpotent powers $\mathbf{R}^k$ and binomial factors, creating finite-order polynomial corrections. Their primary impact is on finite-horizon and boundary behavior before the asymptotic rate dominates. At $f_0=1$, zero-input deviations persist and may grow polynomially.

A caution accompanies Proposition~\ref{prop:open-spectrum}: $|f_0|<1$ guarantees asymptotic stability, not good finite-horizon performance, because the polynomial factors in \eqref{eq:open-binomial} may still generate large transients. Proposition~\ref{prop:transient-envelope} in Section~\ref{sec:finite-boundary-discussion} quantifies these transient effects and their implications for finite-horizon open-chain performance.

\subsection{Computable Steady-State Variance Bounds}\label{subsec:open-variance-bounds}

Assume zero-mean disturbances that are uncorrelated across control points, have fixed finite covariance $\mathbf{Q}:=\Var(\mathbf{v}_{s})$, and are uncorrelated with the initial state, whose covariance is finite. With $\boldsymbol{\Sigma}_{s}:=\Var(\boldsymbol{\varepsilon}_{s})$, the covariance recursion is $\boldsymbol{\Sigma}_{s+1}=\mathbf{A}_{\mathrm{o}}\boldsymbol{\Sigma}_{s}\mathbf{A}_{\mathrm{o}}^{\top}+\mathbf{Q}$. When $|f_0|<1$, the covariances remain bounded and converge to the limit in Proposition~\ref{prop:open-cov} (proof in Appendix~\ref{app:open-cov}).

\begin{proposition}\label{prop:open-cov}
If $|f_0|<1$, the covariance sequence converges to the unique positive semidefinite solution of the discrete Lyapunov equation
\begin{equation}
\boldsymbol{\Sigma}_{\infty}
=\mathbf{A}_{\mathrm{o}}\boldsymbol{\Sigma}_{\infty}\mathbf{A}_{\mathrm{o}}^{\top}+\mathbf{Q}
=\sum_{j=0}^{\infty}\mathbf{A}_{\mathrm{o}}^{j}\mathbf{Q}(\mathbf{A}_{\mathrm{o}}^{\top})^{j}.
\label{eq:open-sigma-infty}
\end{equation}
\end{proposition}

Although \eqref{eq:open-sigma-infty} contains an infinite series, it is computable in finite dimension by solving the discrete Lyapunov equation directly, or by truncating the series once the terms $\mathbf{A}_{\mathrm{o}}^{j}\mathbf{Q}(\mathbf{A}_{\mathrm{o}}^{\top})^{j}$ become negligible.

Schedule- and headway-deviation variances measure how well the control law suppresses service irregularity. We use holding-time variance as a proxy for holding cost, motivated by the three-sigma slack rule \citep{daganzo2009}.

The following corollary translates $\boldsymbol{\Sigma}_{\infty}$ into bounds on schedule-deviation variance, headway-deviation variance, and holding-time standard deviation (proof in Appendix~\ref{app:open-output-bounds}). The headway-deviation vector is $\boldsymbol{\delta}_{s}=(\delta_{2,s},\ldots,\delta_{N,s})^{\top}=\mathbf{G}_{\mathrm{o}}\boldsymbol{\varepsilon}_{s}$, where $\mathbf{G}_{\mathrm{o}}=\left[\begin{smallmatrix}-1&1&0&\cdots&0\\0&-1&1&\ddots&\vdots\\\vdots&\ddots&\ddots&\ddots&0\\0&\cdots&0&-1&1\end{smallmatrix}\right]_{(N-1)\times N}$. The holding-time vector is likewise a linear output of the state, $\mathbf{D}_{s}=d_s\mathbf{1}+\mathbf{M}_{\mathrm{o}}\boldsymbol{\varepsilon}_{s}$, where $\mathbf{1}$ is the $N$-dimensional all-ones vector (different from $\mathbf{I}$) and $\mathbf{M}_{\mathrm{o}}:=-(1+\beta)\mathbf{I}+\beta\mathbf{L}+\mathbf{A}_{\mathrm{o}}$. Since $d_s$ is deterministic, the variance of $\mathbf{D}_{s}$ is determined entirely by $\mathbf{M}_{\mathrm{o}}\boldsymbol{\varepsilon}_{s}$.

\begin{corollary}[Steady-state performance and holding-cost bounds]\label{cor:open-output-bounds}
Suppose $|f_0|<1$ and the open-chain system starts from zero initial deviation. Let $\boldsymbol{\Sigma}_{\infty}$ be defined by Proposition~\ref{prop:open-cov}, let $\mathbf{e}_n\in\mathbb R^N$ be the $n$-th Euclidean basis vector, and define $\mathbf{g}_n:=\mathbf{e}_n-\mathbf{e}_{n-1}$ for $n\ge2$ and $\mathbf{m}_n:=\mathbf{M}_{\mathrm{o}}^{\top}\mathbf{e}_n$. The corresponding steady-state bounds are
\begingroup

\interdisplaylinepenalty=10000
\begin{alignat}{1}
\lim_{s\to\infty}\Var(\varepsilon_{n,s})
&=\mathbf{e}_n^{\top}\boldsymbol{\Sigma}_{\infty}\mathbf{e}_n
\le \mathcal{B}_{\mathrm{o}}, \label{eq:open-schedule-var}\\
\lim_{s\to\infty}\Var(\delta_{n,s})
&=\mathbf{g}_n^{\top}\boldsymbol{\Sigma}_{\infty}\mathbf{g}_n
\le 2\mathcal{B}_{\mathrm{o}}, \label{eq:open-headway-var}\\
\lim_{s\to\infty}\sqrt{\Var(D_{n,s})}
&=\sqrt{\mathbf{m}_n^{\top}\boldsymbol{\Sigma}_{\infty}\mathbf{m}_n}
\le \lVert\mathbf{m}_n\rVert_2\sqrt{\mathcal{B}_{\mathrm{o}}}, \label{eq:open-holding-sd}\\
\text{where}\quad \mathcal{B}_{\mathrm{o}}
&:=
\lVert\mathbf{Q}\rVert_2
\sum_{j=0}^{\infty}
\left(
\sum_{k=0}^{\min\{j,\nu-1\}}
\binom{j}{k}|f_0|^{j-k}\lVert\mathbf{R}\rVert_2^k
\right)^2,\notag
\end{alignat}
\endgroup
and $\nu$ is the nilpotency index in \eqref{eq:open-binomial}, given explicitly in Appendix~\ref{app:supp-nilpotency-index}.
Under the common independent homoskedastic white-noise setting, $\mathbf{Q}=\sigma^2\mathbf{I}$ and hence $\lVert\mathbf{Q}\rVert_2=\sigma^2$. For independent disturbances with heterogeneous bus-specific variances, $\mathbf{Q}=\diag(\sigma_1^2,\ldots,\sigma_N^2)$ and $\lVert\mathbf{Q}\rVert_2=\max_n\sigma_n^2$. Under the same zero-initial-condition white-noise setting, finite-step variances and holding-time standard deviations do not exceed these steady-state limits, so the bounds also apply over finite horizons.
\end{corollary}

The bounds in Corollary~\ref{cor:open-output-bounds} provide conservative estimates of service variability and holding cost. A related fleet-level bound follows from the headway difference matrix (proof in Appendix~\ref{app:open-schedule-headway-norm}).

\begin{corollary}[Fleet-level schedule--headway variance bound]\label{cor:open-schedule-headway-norm}
For $N\ge2$ and any control point \(s\), define the fleet-level schedule-deviation variance $V_{\varepsilon,s}^{\mathrm{o}}:=\tr(\boldsymbol{\Sigma}_{s})/N$ and headway-deviation variance $V_{\delta,s}^{\mathrm{o}}:=\tr(\mathbf{G}_{\mathrm{o}}\boldsymbol{\Sigma}_{s}\mathbf{G}_{\mathrm{o}}^{\top})/(N-1)$. With $\gamma_N:=\frac{N}{N-1}[2+2\cos(\pi/N)]$,
\begin{equation}
0
\le
V_{\delta,s}^{\mathrm{o}}
\le
\gamma_N V_{\varepsilon,s}^{\mathrm{o}}.
\label{eq:open-schedule-headway-bound}
\end{equation}
The same bound holds for the steady-state limits whenever they exist.
\end{corollary}

Bounded schedule-deviation variance therefore implies bounded headway variance. The bound gives no converse implication, since headway differences remove common schedule shifts.

\subsection{\texorpdfstring{Closed-Form Open-Chain Analysis under SC}{Closed-Form Open-Chain Analysis under SC}}
\label{subsec:simple-open-analysis}\label{sec:simple}

The preceding leader-only analysis gives computable steady-state variance bounds, while metrics at individual control points generally require matrix calculations. We therefore specialize the finite-bus open-chain model to one-parameter SC for two reasons. First, SC isolates the only coefficient governing Schur stability within the leader-only open-chain class (Proposition~\ref{prop:open-spectrum}). Second, the scalar recursion generated by \(\mathbf{A}_{\mathrm{o}}=f_0\mathbf{I}\) yields exact schedule, headway, and holding-time variance metrics at each control point and in steady state. These formulas allow direct comparison with the loop counterparts in Section~\ref{sec:loop}.

Section~\ref{subsubsec:simple-point-metrics} derives closed-form performance and holding-cost metrics at each control point, and Section~\ref{subsubsec:simple-parameter-selection} uses their steady-state limits to obtain a closed-form reliability-constrained SC parameter.

\subsubsection{Closed-Form Performance and Holding Cost at Each Control Point}\label{subsubsec:simple-point-metrics}

Under SC, each bus follows the scalar recursion $\varepsilon_{n,s+1}=f_0\varepsilon_{n,s}+v_{n,s+1}$. Under zero initial deviation and independent homoskedastic disturbances, this recursion gives exact second moments at each control point \(s\) (proof in Appendix~\ref{app:simple-point-metrics}).

\begin{proposition}[Metrics at each control point under SC]\label{prop:simple-point-metrics}
Assume $N\ge2$, \(0\le f_0<1\), zero initial deviation, and independent disturbances across buses and control points with mean zero and variance \(\sigma^2\). For any control point \(s\ge1\), the normalized schedule-deviation, fleet-average headway-deviation, and interior-bus (\(n\ge2\)) holding-time variances, respectively, are
\begingroup

\interdisplaylinepenalty=10000
\begin{alignat}{1}
J_{\varepsilon,s}(f_0)
&:=
\frac{\Var(\varepsilon_{n,s})}{\sigma^2}
=
\frac{1-f_0^{2s}}{1-f_0^2},
\label{eq:simple-Jepsilon}\\
J_{\delta,s}(f_0)
&:=
\frac{1}{(N-1)\sigma^2}\tr(\mathbf{G}_{\mathrm{o}}\boldsymbol{\Sigma}_{s}\mathbf{G}_{\mathrm{o}}^{\top})
=
2J_{\varepsilon,s}(f_0),
\label{eq:simple-Jdelta}\\
J_{D,s}(f_0;\beta)
&:=
\frac{\Var(D_{n,s})}{\sigma^2}
=
B(f_0;\beta)J_{\varepsilon,s}(f_0),
\label{eq:simple-JD}
\end{alignat}
\endgroup
where $B(f_0;\beta):=(1+\beta-f_0)^2+\beta^2$.
\end{proposition}

Under a three-sigma slack rule, minimizing the normalized holding-time variance \(J_{D,s}(f_0;\beta)\) also minimizes the corresponding holding-time standard deviation and slack. In this one-parameter case, independent homogeneous disturbances make \(\boldsymbol{\Sigma}_{s}\) a scalar multiple of \(\mathbf I\), and the three metrics in Proposition~\ref{prop:simple-point-metrics} are fixed multiples of \(J_{\varepsilon,s}(f_0)\) as \(s\) varies. These proportional relationships cannot be assumed for general kernels.

For fixed \(f_0\in[0,1)\), these three metrics are nondecreasing in \(s\) and converge to their steady-state limits. The next subsection uses these limits to obtain a compact closed-form optimal design.

\subsubsection{Optimal SC}\label{subsubsec:simple-parameter-selection}

Proposition~\ref{prop:simple-point-metrics} directly supports finite-\(s\) parameter selection. Given schedule- and headway-performance thresholds \(\bar J_{\varepsilon}\) and \(\bar J_{\delta}\), a design problem at control point \(s\) is
\begin{equation}
\begin{aligned}
\inf_{0\le f_0<1}\quad
& J_{D,s}(f_0;\beta) \\
\text{s.t.}\quad
& J_{\varepsilon,s}(f_0)\le \bar J_{\varepsilon},\qquad
J_{\delta,s}(f_0)\le \bar J_{\delta}.
\end{aligned}
\label{eq:simple-point-opt}
\end{equation}
Here the objective is the normalized holding-time variance of an interior bus ($n\ge2$) at the control point. Since $J_{\delta,s}(f_0)=2J_{\varepsilon,s}(f_0)$, the two reliability constraints reduce to upper bounds on the same scalar function \(J_{\varepsilon,s}(f_0)\).

For finite \(s\), the one-dimensional problem need not attain its infimum and generally lacks a compact closed-form solution. We therefore focus on the steady-state limit to obtain a closed-form design (proof in Appendix~\ref{app:simple-asymptotic-opt}).

\begin{corollary}[Steady-state optimal SC]\label{cor:simple-asymptotic-opt}
Suppose the steady-state thresholds satisfy \(\bar J_{\varepsilon}\ge1\) and \(\bar J_{\delta}\ge2\). As \(s\to\infty\),
\begingroup

\begin{equation}
J_{\varepsilon,\infty}(f_0)=\frac{1}{1-f_0^2},\qquad
J_{\delta,\infty}(f_0)=\frac{2}{1-f_0^2},\qquad
J_{D,\infty}(f_0;\beta)=\frac{(1+\beta-f_0)^2+\beta^2}{1-f_0^2}.
\label{eq:simple-asym-metrics}
\end{equation}
\endgroup
Consider the steady-state constrained problem
\begin{equation}
\begin{aligned}
\min_{0\le f_0<1}\quad
& J_{D,\infty}(f_0;\beta)\\
\text{s.t.}\quad
& J_{\varepsilon,\infty}(f_0)\le \bar J_{\varepsilon},\qquad
J_{\delta,\infty}(f_0)\le \bar J_{\delta}.
\end{aligned}
\label{eq:asym-constraint}
\end{equation}
Define
\begin{equation}
f_0^*(\beta)
=
\frac{\beta^2+\beta+1-\beta\sqrt{\beta^2+2\beta+2}}{1+\beta},
\label{eq:f0star}
\end{equation}
with reliability bounds $f_{\max,\varepsilon}^{\infty}=\sqrt{1-1/\bar J_{\varepsilon}}$ and $f_{\max,\delta}^{\infty}=\sqrt{1-2/\bar J_{\delta}}$. The optimal steady-state parameter is
\begin{equation}
f_0^{\mathrm{opt},\infty}
=
\min\{f_0^*(\beta),f_{\max,\varepsilon}^{\infty},f_{\max,\delta}^{\infty}\}.
\label{eq:simple-asym-opt}
\end{equation}
\end{corollary}

The steady-state problem \eqref{eq:asym-constraint} uses the reliability-constrained SC formulation of \citet{xuan2011}. Operationally, the reliability thresholds set the admissible range for \(f_0\), and the holding-time variance selects the cost-efficient point within it. Increasing \(f_0\) reduces holding-time variance up to its minimizer \(f_0^*(\beta)\), at the expense of schedule and headway reliability. Beyond this minimizer, both holding-time variance and reliability worsen, so these choices are Pareto dominated.

{
\subsection{Open Chains with Following-Bus Feedback}\label{subsec:open-two-sided}

The preceding analysis focuses on current- and leading-bus feedback. When feedback from both leading and following buses is used, the resulting transition matrix is not triangular. In this section, we consider a simple case of $N\ge2$ buses under a real three-tap kernel with offsets $\{-1,0,1\}$. Missing terms are set to zero at both ends. Recall that any required but unobserved deviation is assumed to be predicted exactly at the decision epoch. The controlled dynamics are
\begin{equation}
\boldsymbol{\varepsilon}_{s+1}
=\mathbf{A}_{\mathrm{o}}^{\pm}\boldsymbol{\varepsilon}_{s}+\mathbf{v}_{s+1},
\qquad
\mathbf{A}_{\mathrm{o}}^{\pm}
=f_0\mathbf{I}+f_1\mathbf{L}+f_{-1}\mathbf{L}^{\top}.
\label{eq:open-two-sided-matrix}
\end{equation}
As $\mathbf{A}_{\mathrm{o}}^{\pm}$ is a tridiagonal Toeplitz matrix \citep{smith1978}, we have the following exact characterization (proof in Appendix~\ref{app:open-two-sided}).

\begin{proposition}[Three-tap open-chain stability]\label{prop:open-two-sided}
Let $c=f_1f_{-1}$ and $\chi_N=\cos(\pi/(N+1))$. The eigenvalues of $\mathbf{A}_{\mathrm{o}}^{\pm}$ are
\begin{equation}
\lambda_k=f_0+2\sqrt{c}\cos\!\left(\frac{k\pi}{N+1}\right),
\qquad k=1,\ldots,N,
\label{eq:open-two-sided-spectrum}
\end{equation}
where $\mathrm{i}^2=-1$ and $\sqrt{c}=\mathrm{i}\sqrt{|c|}$ when $c<0$. The system is Schur stable if and only if
\begin{equation}
\begin{cases}
|f_0|+2\sqrt{c}\,\chi_N<1, & c>0,\\
f_0^2+4|c|\chi_N^2<1, & c<0,\\
|f_0|<1, & c=0.
\end{cases}
\label{eq:open-two-sided-stability}
\end{equation}
\end{proposition}

When both neighboring-bus coefficients are nonzero, stability depends on their product $c$ as well as $f_0$. Thus, the eigenvalues all equal $f_0$ because the feedback is one-sided on an open chain; the open boundary alone does not ensure this. If either neighboring-bus coefficient vanishes, the matrix becomes triangular and the condition reduces to $|f_0|<1$. With both present, information can circulate through adjacent buses even though the chain remains open. Different coefficient ratios with the same product leave the spectrum unchanged but can change transient and covariance behavior.

Under the same white-noise assumptions, performance can still be computed from the covariance recursion and Lyapunov equation \eqref{eq:open-sigma-infty}, replacing $\mathbf{A}_{\mathrm{o}}$ with $\mathbf{A}_{\mathrm{o}}^{\pm}$; steady-state covariances exist when \eqref{eq:open-two-sided-stability} holds. The headway difference matrix $\mathbf{G}_{\mathrm{o}}$ is unchanged, so Corollary~\ref{cor:open-schedule-headway-norm} still applies.

Recall that for leader-only kernels, triangularity gives the $f_0$-only stability criterion, while nilpotency yields the truncated expansion \eqref{eq:open-binomial} and variance bound $\mathcal{B}_{\mathrm{o}}$. When both $f_1$ and $f_{-1}$ are nonzero, $\mathbf{A}_{\mathrm{o}}^{\pm}-f_0\mathbf{I}$ is no longer nilpotent, so this truncation and bound no longer apply; transients and variances can instead be computed directly from matrix powers and covariance equations. Wider two-sided kernels generally lack the same structural simplifications. Their eigenvalues and covariances remain computable numerically, but general closed-form analysis and optimal design become intricate.
}

\section{Periodic-Fleet Stability and Performance on a Loop Corridor}\label{sec:loop}

On a loop with $N$ buses and $S$ physical control points, the scheduled circuit time must equal $NH$, so additional planned slack increases $H$ at fixed fleet size and operating conditions. Once a feasible deterministic schedule is fixed, $H$ drops out of \eqref{eq:local-closed}, leaving the deviation analysis unaffected.

Starting from the same nonnegative-offset kernel expression as in the leader-only open-chain model, this section develops the general cyclic-kernel theory. Section~\ref{subsec:loop-modal-setup} constructs the circulant representation and its Fourier modes; the exact stability condition and its relation to the bi-infinite system follow; Section~\ref{subsec:loop-performance-metrics} derives variances for schedule and headway deviations, as well as holding time; and Section~\ref{subsec:loop-modal-optimization} solves the corresponding optimal design problems.

\subsection{Loop Corridor Matrix Representation and Modal Decomposition}\label{subsec:loop-modal-setup}

Consider a loop corridor with $N$ physical buses circulating indefinitely, or for a specified number of completed laps. We use $s$ as the unrolled control-point index, with $S\ge1$ control points per circuit. At each control point, stack the schedule deviations across buses as $\boldsymbol{\varepsilon}_{s}=(\varepsilon_{1,s},\varepsilon_{2,s},\ldots,\varepsilon_{N,s})^{\top}$.
For kernel inputs, the bus index is periodic, so $n\equiv n+N$. Define the cyclic shift matrix $\mathbf{P}=\left[\begin{smallmatrix}
0&0&\cdots&1\\
1&0&\cdots&0\\
&\ddots&\ddots&\vdots\\
0&&1&0
\end{smallmatrix}\right]_{N\times N}$, so that $(\mathbf{P}\boldsymbol{\varepsilon}_{s})_n=\varepsilon_{n-1,s}$ with $\varepsilon_{0,s}\equiv \varepsilon_{N,s}$.
The cyclic shift $\mathbf{P}$ differs from the open-chain shift $\mathbf{L}$ only in the top-right entry: a single zero is replaced by a one that closes the chain. This looks like a small boundary change, but it is algebraically decisive. The open-chain shift is nilpotent, $\mathbf{L}^{N}=\mathbf{0}$, whereas the loop shift is periodic, satisfying $\mathbf{P}^{N}=\mathbf{I}$; this entry converts the leader-only open-chain transition matrix into a circulant matrix, whose successive rows are cyclic shifts of one another. It also changes the stability condition, as detailed next.

Periodic indexing in \eqref{eq:local-closed} gives the loop state dynamics
\begin{equation}
\boldsymbol{\varepsilon}_{s+1}
=
\mathbf{A}_{\ell}\boldsymbol{\varepsilon}_{s}
+
\mathbf{v}_{s+1},
\label{eq:loop-dynamics}
\end{equation}
where $\mathbf{A}_{\ell}=\sum_{p=0}^{m} f_p\mathbf{P}^{p}$.
{
For fixed coefficients $f_p$, replacing $\mathbf L$ by $\mathbf P$ maps the truncated open-chain matrix to the loop matrix.

Under periodic indexing, nonnegative offsets suffice to represent any finite two-sided kernel. Since $\mathbf P^N=\mathbf I$, in particular $\mathbf P^{-1}=\mathbf P^{N-1}$, any finite two-sided kernel can be represented with offsets $0,\ldots,N-1$ by summing coefficients whose offsets agree modulo $N$. We therefore develop the full cyclic class below.
}

The cyclic boundary yields the following Fourier decomposition (proof in Appendix~\ref{app:loop-modal-decomposition}).
Fourier analysis was also used for two-way speed control by \citet[Appendix~A]{daganzo2011}, but not for the general cyclic-kernel optimization developed here.
Define $\theta_q=2\pi q/N$ and $z_q=e^{-\mathrm{i}\theta_q}$ for $q=0,1,\ldots,N-1$, equally spaced on the unit circle. Let $\mathbf{u}_q$ be the Fourier vector with entries $(\mathbf{u}_q)_n=e^{\mathrm{i}(n-1)\theta_q}/\sqrt{N}$ for $n=1,\ldots,N$, and let $\mathbf{U}=[\mathbf{u}_0,\mathbf{u}_1,\ldots,\mathbf{u}_{N-1}]$.
For vectors and matrices, the superscript $^{*}$ denotes the conjugate transpose; on design variables, it denotes an optimum.

\begin{proposition}[Finite-loop modal decomposition]\label{prop:loop-modal-decomposition}
For $q=0,1,\ldots,N-1$, define
\begin{equation}
\lambda_q=F(z_q)=\sum_{p=0}^{m} f_p e^{-\mathrm{i}p\theta_q},
\qquad F(z):=\sum_{p=0}^{m} f_p z^p.
\label{eq:loop-symbol}
\end{equation}
Then, with $\widehat{\boldsymbol{\varepsilon}}_s=\mathbf{U}^{*}\boldsymbol{\varepsilon}_s$ and $\widehat{\mathbf v}_{s+1}=\mathbf{U}^{*}\mathbf v_{s+1}$, the transformed loop dynamics decompose into $N$ decoupled scalar recursions,
\begin{equation}
\widehat{\varepsilon}_{q,s+1}
=
\lambda_q\widehat{\varepsilon}_{q,s}
+
\widehat v_{q,s+1},
\qquad q=0,1,\ldots,N-1.
\label{eq:modal-dynamics}
\end{equation}
\end{proposition}

Proposition~\ref{prop:loop-modal-decomposition} marks the algebraic departure from the leader-only finite open chain: the triangular transition matrix based on a nilpotent shift gives way to a circulant matrix diagonalized by the Fourier basis, so the loop decomposes into scalar modal recursions. This is why our loop analysis proceeds through Fourier modes rather than the Lyapunov bounds, and it yields the exact finite-mode variance formulas in Section~\ref{subsec:loop-performance-metrics}.

The modal coordinate $\widehat{\varepsilon}_{q,s}=\mathbf{u}_q^{*}\boldsymbol{\varepsilon}_s=(1/\sqrt N)\sum_{n=1}^{N}e^{-\mathrm{i}(n-1)\theta_q}\varepsilon_{n,s}$ is not any one bus's deviation but the amplitude of the $q$-th Fourier pattern across the fleet at a control point, and $\lambda_q$ is its one-step multiplier. Spatial frequency describes variation across the ordered buses. The magnitude $|\lambda_q|$, termed the modal gain, determines the pattern's attenuation or amplification at each step. The left and right endpoints on the unit circle are particularly informative. At $q=0$ ($\theta_0=0$), $\widehat{\varepsilon}_{0,s}=(1/\sqrt N)\sum_n\varepsilon_{n,s}$ is $\sqrt N$ times the fleet-average schedule deviation. For modes near $\theta_q=\pi$ (near $q=N/2$), adjacent buses differ in phase by $e^{-\mathrm{i}\theta_q}\approx-1$ and are thus nearly opposite; these adjacent-contrast modes govern headway regularity, since $h_{n,s}-H=\varepsilon_{n,s}-\varepsilon_{n-1,s}$. Choosing the modal multipliers $\lambda_q=F(z_q)$ therefore controls how strongly each pattern is damped.

\subsection{Loop Stability and Relation to the Bi-Infinite System}\label{subsec:loop-stability-xuan}

The modal decomposition in Proposition~\ref{prop:loop-modal-decomposition} gives the multipliers $\lambda_q=F(z_q)$, which determine whether spatial patterns of schedule deviations decay to zero in the absence of disturbances. This yields the following exact stability criterion.

\begin{proposition}[Finite-loop stability]\label{prop:loop-stability}
The zero-input finite loop system is Schur stable if and only if
\begin{equation}
\max_{q=0,\ldots,N-1}|F(z_q)|<1.
\label{eq:loop-stability}
\end{equation}
\end{proposition}

The proposition shows that loop stability depends on the gains $|F(z_q)|$ of all Fourier modes. This differs from the leader-only open-chain condition $|f_0|<1$, and is similar to the two-sided open-chain case (Proposition~\ref{prop:open-two-sided}). In a periodic fleet, coupling coefficients can create or remove unstable spatial modes. A kernel stable under open-chain truncation can therefore violate Schur stability on a loop.

Conversely, for a kernel using only deviations from buses $n,n-1,\ldots,n-m$ with $m<N$, loop stability implies stability of its leader-only open-chain implementation, since $|f_0|=|N^{-1}\sum_q\lambda_q|\le\max_q|\lambda_q|<1$. This implication can fail with following-bus feedback: a two-sided kernel can be stable on the loop but unstable on the open chain.

{
For example, take $N=3$, with all required deviations available exactly in both models. For $(f_0,f_1,f_{-1})=(0.5,0.3,0.3)$, Proposition~\ref{prop:open-two-sided} gives an open-chain spectral radius of $0.5+0.6\cos(\pi/4)\approx0.924<1$, whereas Proposition~\ref{prop:loop-stability} gives a loop spectral radius of $1.1>1$. The kernel $(0.4,-0.5,-0.5)$ reverses the outcome: the spectral radii become approximately $1.107>1$ on the open chain and $0.9<1$ on the loop. Thus neither stability region contains the other in general, even with the same three-tap kernel and exact information. (Cross-lap feedback, introduced below, more closely matches the open chain's physical-neighbor inputs; for the same kernels with $S=3$, Appendix~\ref{app:seam-delay} gives loop spectral radii of $1.103$ and $0.945$, respectively, preserving both stability contrasts found under same-lap feedback.)
}

{
\begin{remark}[Same-lap and cross-lap feedback]\label{rem:seam}
The loop stability result uses the same-lap inputs in \eqref{eq:loop-dynamics}: bus deviations at the same physical control point on the same numbered lap, all indexed by $s$. At the connection between buses $N$ and $1$ (the seam), the $f_1$ term for bus $1$ therefore uses the not-yet-observed $\varepsilon_{N,s}$. Using the observed preceding-lap value $\varepsilon_{N,s-S}$ instead gives a cross-lap feedback law, which can be modeled with a larger state containing both current and past deviations. Appendix~\ref{app:seam-delay} develops this model and gives a two-tap example that is stable on the open chain but unstable on the loop under either feedback convention; its loop instability is therefore not an artifact of the same-lap assumption. SC remains a special case on the loop: it has no wrapping kernel input, so the same- and cross-lap feedback laws coincide.
\end{remark}
}

As $N\to\infty$, with the coefficients and offsets of $F$ fixed, the loop's Fourier sampling points become dense on the unit circle, so its maximum modal gain converges to the supremum over all spatial frequencies. For the bi-infinite system underlying the analysis of \citet{xuan2011}, the condition
\begin{equation}
\sup_{\theta\in[0,2\pi)}
|F(e^{-\mathrm{i}\theta})|
<1
\label{eq:infinite-fourier}
\end{equation}
is necessary and sufficient for uniform geometric decay of all spatial Fourier modes in the absence of disturbances.

\citet{xuan2011} analyze a convolution-type linear control family and use a sufficient boundedness condition of the form
\begin{equation}
\sum_{p=0}^{m}|f_p|<1.
\label{eq:xuan-sufficient}
\end{equation}
By the triangle inequality, $|F(e^{-\mathrm{i}\theta})|=|\sum_{p=0}^{m}f_p e^{-\mathrm{i}p\theta}|\le\sum_{p=0}^{m}|f_p|$.
Hence \eqref{eq:xuan-sufficient} is sufficient for \eqref{eq:infinite-fourier}. The two conditions coincide when all coefficients are nonnegative, since equality holds at $\theta=0$. With mixed-sign coefficients, cancellation can make the coefficient-sum condition \eqref{eq:xuan-sufficient} conservative. Section~\ref{sec:numerics} quantifies the magnitude of this conservatism and its growth as deviations from more buses are allowed in the kernel.

\subsection{Closed-Form Loop Performance and Cost Metrics}\label{subsec:loop-performance-metrics}\label{subsec:loop-headway-variance}

For the same-lap dynamics \eqref{eq:loop-dynamics}, define the same-lap outputs
\begingroup
\[
\begin{aligned}
\bar{\boldsymbol\delta}_s&=\mathbf G_\ell\boldsymbol\varepsilon_s,
&\mathbf G_\ell&:=\mathbf I-\mathbf P,\\
\bar{\mathbf D}_s&=d_s\mathbf1+\mathbf M_\ell\boldsymbol\varepsilon_s,
&\mathbf M_\ell&:=\mathbf A_\ell-(1+\beta)\mathbf I+\beta\mathbf P.
\end{aligned}
\]
\endgroup
The matrices $\mathbf G_\ell$ and $\mathbf M_\ell$ are the loop counterparts of $\mathbf G_{\mathrm o}$ and $\mathbf M_{\mathrm o}$ in Section~\ref{subsec:open-variance-bounds}. The normalized fleet-average output variances are $J_{\delta,s}^{\ell}$ and $J_{D,s}^{\ell}$; $J_{\varepsilon,s}^{\ell}$ is the schedule-variance metric. Fourier diagonalization gives the following formulas (proof in Appendix~\ref{app:loop-performance-metrics}).

\begin{proposition}[Loop performance metrics]\label{prop:loop-performance-metrics}
Assume zero initial deviations, $\mathbb{E}[\mathbf{v}_{s}]=\mathbf{0}$, $\Var(\mathbf{v}_{s})=\sigma^2\mathbf{I}$, and independence across control points. Define $g_q=(1+\beta)-\beta z_q$, $h_q=|1-z_q|^2$, and $\mu_q:=\lambda_q-g_q$.
After $s$ passages, the metrics are as follows, with each term interpreted by continuity when $|\lambda_q|=1$:
\begingroup

\interdisplaylinepenalty=10000
\begin{alignat}{1}
J_{\varepsilon,s}^{\ell}
&:=\frac{1}{N\sigma^2}\tr(\Var(\boldsymbol{\varepsilon}_{s}))
=\frac{1}{N}\sum_{q=0}^{N-1}\frac{1-|\lambda_q|^{2s}}{1-|\lambda_q|^2},
\label{eq:loop-schedule-finite}\\
J_{\delta,s}^{\ell}
&:=\frac{1}{N\sigma^2}\tr(\Var(\bar{\boldsymbol{\delta}}_{s}))=\frac{1}{N}\sum_{q=0}^{N-1}h_q\frac{1-|\lambda_q|^{2s}}{1-|\lambda_q|^2},
\label{eq:loop-headway-finite}\\
J_{D,s}^{\ell}
&:=\frac{1}{N\sigma^2}\tr(\Var(\bar{\mathbf{D}}_{s}))
=\frac{1}{N}\sum_{q=0}^{N-1}|\mu_q|^2\frac{1-|\lambda_q|^{2s}}{1-|\lambda_q|^2}.
\label{eq:loop-holding-finite}
\end{alignat}
\endgroup
Under the strict stability condition \eqref{eq:loop-stability}, the steady-state limits are
\begingroup

\interdisplaylinepenalty=10000
\begin{alignat}{1}
J_{\varepsilon,\infty}^{\ell}
&=\frac{1}{N}\sum_{q=0}^{N-1}\frac{1}{1-|\lambda_q|^2},
\label{eq:loop-schedule-ss}\\
J_{\delta,\infty}^{\ell}
&=\frac{1}{N}\sum_{q=0}^{N-1}\frac{h_q}{1-|\lambda_q|^2},
\label{eq:loop-headway-ss}\\
J_{D,\infty}^{\ell}
&=\frac{1}{N}\sum_{q=0}^{N-1}\frac{|\mu_q|^2}{1-|\lambda_q|^2}.
\label{eq:loop-holding-ss}
\end{alignat}
\endgroup
\end{proposition}

The holding outputs share variance $\sigma^2J_{D,\infty}^{\ell}(F)$, with three-sigma scale $3\sigma\sqrt{J_{D,\infty}^{\ell}(F)}$.

The above variance formulas are additive across modes. In steady state, the same modal variance-accumulation factor $1/(1-|\lambda_q|^2)$ enters the schedule, headway, and holding metrics for each mode with output weights $1$, $h_q$, and $|\mu_q|^2$, respectively. These output weights distinguish schedule anchoring from headway smoothing. The schedule metric assigns equal output weight to every mode, whereas the headway weight $h_q=|1-z_q|^2=2-2\cos\theta_q=4\sin^2(\theta_q/2)$ is zero at $\theta_q=0$ and rises to its maximum $4$ at $\theta_q=\pi$. Thus, the headway metric drops the fleet-average mode $q=0$ and emphasizes the adjacent-contrast modes near $\theta_q=\pi$---the short--long headway alternation around the loop. A controller that effectively curbs schedule deviations can therefore still perform poorly in terms of headway regularity if it leaves those adjacent-contrast modes weakly damped; to our knowledge, this mode-based distinction has not been adequately explained and verified in the literature. The holding weight $|\mu_q|^2$ measures how strongly the control law translates each mode's accumulated schedule variance into holding-time variance.

Since $0\le h_q\le4$, termwise comparison of \eqref{eq:loop-schedule-finite} and \eqref{eq:loop-headway-finite} yields the following bound between the fleet-average schedule and headway variances on a loop. This one-sided relation is consistent with \citet{xuan2011} and the open-chain bound in Corollary~\ref{cor:open-schedule-headway-norm}.

\begin{corollary}[Loop schedule--headway bound]\label{cor:loop-schedule-headway-upper}
For every finite number \(s\) of control-point passages,
\begin{equation}
0
\le
J_{\delta,s}^{\ell}
\le
4J_{\varepsilon,s}^{\ell}.
\label{eq:loop-schedule-headway-upper}
\end{equation}
Thus, on the loop, bounded schedule-deviation variance implies bounded headway variance, but not conversely. In particular, if all nonzero Fourier modes are stable but the zero-frequency mode is not, the headway-deviation metric can remain bounded while the schedule-deviation metric diverges.
\end{corollary}

{
\begin{remark}[Cross-lap corrections to headway and holding outputs]\label{rem:physical-outputs}
For analytical tractability, the headway and holding formulas above approximate bus~$1$'s preceding-lap leader deviation $\varepsilon_{N,s-S}$ by the same-lap value $\varepsilon_{N,s}$. The headway to its actual leader is $h_{1,s}=H+\varepsilon_{1,s}-\varepsilon_{N,s-S}$. \textbf{Cross-lap corrections} account for bus~$1$'s preceding-lap leader in the headway output and the demand-compensation term of the holding law, while retaining same-lap kernel inputs. Note that this is different from the \textbf{cross-lap feedback law} in Remark~\ref{rem:seam}. Demand cancellation remains exact and the state dynamics are unchanged; only bus~$1$'s contributions to the headway and holding variances in \eqref{eq:loop-headway-finite}--\eqref{eq:loop-holding-finite} require adjustment. Fleet averaging reduces each resulting variance discrepancy to $1/N$ of its bus~$1$ value. Appendices~\ref{app:physical-metrics}--\ref{app:physical-bounds} provide exact variance corrections and conditions under which budget tightening yields feasible analytical designs with bounded excess holding cost; Section~\ref{subsec:num-frontiers} shows that the analytical design, retuned using exact corrections, remains close to the full cyclic numerical frontier at the tested budgets.
\end{remark}
}

Under SC, cross-lap corrections leave the steady-state variances unchanged, and the loop metrics specialize cleanly. Setting $F(z)=f_0$ gives $\lambda_q=f_0$ for every mode, and Proposition~\ref{prop:loop-performance-metrics}, using $\sum_{q}z_q=0$, yields for $N>1$ and $|f_0|<1$ the steady-state values $J_{\varepsilon,\infty}^{\ell}=1/(1-f_0^2)$, $J_{\delta,\infty}^{\ell}=2/(1-f_0^2)$, and $J_{D,\infty}^{\ell}=[(1+\beta-f_0)^2+\beta^2]/(1-f_0^2)$. These coincide with the open-chain SC schedule and headway formulas in Section~\ref{subsec:simple-open-analysis}; the holding formula coincides with that for interior open-chain buses.

\subsubsection{Headway-Based Control}\label{subsubsec:loop-headway-based-control}
We call a cyclic kernel headway-based when its holding output depends only on same-lap adjacent differences:
{
\begin{equation}
\bar{\mathbf{D}}_{s}-d_s\mathbf{1}
=
K(\mathbf{P})\mathbf G_\ell
\boldsymbol{\varepsilon}_{s},
\qquad
K(z)=\sum_{p=0}^{m_h}a_p z^p ,
\label{eq:headway-based-holding}
\end{equation}
}
Here $\mathbf G_\ell\boldsymbol\varepsilon_s=\bar{\boldsymbol\delta}_s$, and $K(\mathbf P)$ forms weighted sums of these headway deviations to determine holding adjustments. This control law has no absolute schedule anchor; accounting for bus~$1$'s preceding-lap leader adds the cross-lap correction in \eqref{eq:physical-outputs} to its holding output (Appendix~\ref{app:physical-metrics}). The resulting state multiplier has the following form (proof in Appendix~\ref{app:headway-based-zero-mode}):

\begin{proposition}[Headway-based feedback leaves the fleet-average schedule mode neutral]\label{prop:headway-based-zero-mode}
For any headway-based cyclic kernel of the form \eqref{eq:headway-based-holding}, the one-step multiplier of the mode with cyclic-shift coordinate \(z\) is
\begin{equation}
F_{\mathrm h}(z)
=
(1+\beta)-\beta z
+
K(z)(1-z).
\label{eq:headway-based-symbol}
\end{equation}
At the zero-frequency coordinate \(z_0=1\), $\lambda_0=F_{\mathrm h}(z_0)=1$.
Consequently, the zero-frequency component, corresponding to the fleet-average schedule deviation, is carried forward unchanged, and the schedule-deviation state is not Schur stable.
\end{proposition}

The following corollary shows that under persistent disturbances, the neutral mode (with $\lambda_0=1$) causes schedule variance to grow (proof in Appendix~\ref{app:headway-based-schedule-divergence}):

\begin{corollary}[Schedule divergence under persistent disturbances]\label{cor:headway-based-schedule-divergence}
Under the initial-state and disturbance assumptions stated in Proposition~\ref{prop:loop-performance-metrics}, any headway-based cyclic kernel satisfying Proposition~\ref{prop:headway-based-zero-mode} has
\begin{equation}
J_{\varepsilon,s}^{\ell}
=
\frac{s}{N}
+
\frac{1}{N}
\sum_{q=1}^{N-1}
\frac{1-|\lambda_q|^{2s}}{1-|\lambda_q|^2}
\ge
\frac{s}{N}.
\label{eq:headway-based-schedule-growth}
\end{equation}
with the same continuity convention as in Proposition~\ref{prop:loop-performance-metrics}. Therefore, the schedule-deviation metric diverges as \(s\to\infty\). If, in addition, $\max_{q=1,\ldots,N-1}|\lambda_q|<1$, then the headway-deviation metric has the finite steady-state limit
\[
J_{\delta,\infty}^{\ell}
=
\frac{1}{N}
\sum_{q=1}^{N-1}
\frac{h_q}{1-|\lambda_q|^2},
\]
while \(J_{\varepsilon,s}^{\ell}=s/N+O(1)\).
\end{corollary}

This is the finite loop counterpart of the kernel-sum-one schedule-divergence result in \citet{xuan2011}: relative modes can remain stable while the common schedule mode is a random walk. When the cross-lap correction to bus~$1$'s headway is included, the neutral common mode contributes the finite constant $S/N^2$ to the normalized fleet-average headway variance; with stable relative modes, the variance therefore still converges (see Appendix~\ref{app:physical-neutral}).

Two specific headway-based control schemes, the BHC and TWHC variants of \citet{xuan2011}, are analyzed in Appendices~\ref{app:bhc} and~\ref{app:twhc}, respectively, demonstrating the value of modal analysis for diagnosing specific control schemes.

\subsection{Modal Constrained Optimization on a Loop}\label{subsec:loop-modal-optimization}

Using the same-lap variance metrics derived above, Fourier decomposition gives an exact finite-dimensional optimization over the full cyclic kernel $F_f(z)=\sum_{p=0}^{N-1}f_pz^p$. This entails no loss of generality for circulant matrices because $\mathbf P^N=\mathbf I$. Fleet-average holding-time variance measures the mean squared adjustment around the scheduled slack and provides an analytically tractable proxy for holding cost. Let $R_\varepsilon$ and $R_\delta$ be feasible finite budgets for the schedule and headway metrics, respectively. The resulting design problem is
\begin{equation}
\begin{aligned}
\min_{f_0,\ldots,f_{N-1}}\quad
&
\frac{1}{N}\sum_{q=0}^{N-1}
\frac{|F_f(z_q)-g_q|^2}{1-|F_f(z_q)|^2}
\\
\mathrm{s.t.}\quad
&
\frac{1}{N}\sum_{q=0}^{N-1}
\frac{1}{1-|F_f(z_q)|^2}
\le R_{\varepsilon},
\\
&
\frac{1}{N}\sum_{q=0}^{N-1}
\frac{h_q}{1-|F_f(z_q)|^2}
\le R_{\delta},
\\
&
|F_f(z_q)|<1,
\qquad q=0,\ldots,N-1 .
\end{aligned}
\label{eq:loop-full-kernel-optimization}
\end{equation}
Equation \eqref{eq:loop-full-kernel-optimization} requires full schedule stability. Let \(\eta_{\varepsilon}\ge0\) and \(\eta_{\delta}\ge0\) denote the Lagrange multipliers associated with the schedule- and headway-reliability constraints. A design without a schedule anchor instead permits a neutral zero mode and is treated in Section~\ref{subsubsec:loop-headway-equalizing}.

The following result characterizes the optimal full cyclic kernel (proof in Appendix~\ref{app:loop-optimal-design}).
\begin{proposition}[Optimal full cyclic control]\label{prop:loop-optimal-design}
{
Let $N\ge2$ and $\beta>0$, and assume $1<R_\varepsilon<\infty$ and $R_\delta>2$ whenever the corresponding constraints are imposed. Define $G_q:=|g_q|=\sqrt{1+\beta(1+\beta)h_q}$ and $W_q:=\eta_\varepsilon+\eta_\delta h_q$. With the schedule constraint imposed, Problem~\eqref{eq:loop-full-kernel-optimization} has a unique real optimal cyclic kernel. Its optimal modal gains $r_q^*=|\lambda_q^*|$ and multipliers $\lambda_q^*$, for $q=0,\ldots,N-1$, are
\begin{equation}
r_q^*(\eta_{\varepsilon},\eta_{\delta})
=
\frac{1+G_q^2+W_q-
\sqrt{(1+G_q^2+W_q)^2-4G_q^2}}
{2G_q}.
\label{eq:modal-rstar}
\end{equation}
\begin{equation}
\lambda_q^*(\eta_{\varepsilon},\eta_{\delta})
=
r_q^*(\eta_{\varepsilon},\eta_{\delta})\frac{g_q}{G_q}.
\label{eq:modal-lambdastar}
\end{equation}
The corresponding coefficients, for $p=0,\ldots,N-1$, are
\begin{equation}
f_p^*(\eta_{\varepsilon},\eta_{\delta})
=\frac{1}{N}\sum_{q=0}^{N-1}
\lambda_q^*(\eta_{\varepsilon},\eta_{\delta})z_q^{-p}.
\label{eq:modal-fstar}
\end{equation}
The multipliers satisfy $\eta_\varepsilon>0$ and $\eta_\delta\ge0$. Evaluated at the optimal kernel, the metrics satisfy feasibility and complementary slackness:
\[
\begin{aligned}
J_{\varepsilon,\infty}^{\ell}&\le R_\varepsilon,
&\qquad \eta_\varepsilon(J_{\varepsilon,\infty}^{\ell}-R_\varepsilon)&=0,\\
J_{\delta,\infty}^{\ell}&\le R_\delta,
&\qquad \eta_\delta(J_{\delta,\infty}^{\ell}-R_\delta)&=0.
\end{aligned}
\]
If the headway constraint is omitted, set $\eta_\delta=0$ and omit its feasibility and slackness conditions.

With only a headway-variance constraint $2<R_\delta<\infty$, allow $\lambda_0\in(-1,1]$ and require all nonzero modes to be strictly stable. The unique optimum has $\lambda_0^*=1$; its remaining modal values follow \eqref{eq:modal-rstar}--\eqref{eq:modal-lambdastar} with $\eta_\varepsilon=0$, and its coefficients follow \eqref{eq:modal-fstar}. Only the headway feasibility and complementary-slackness conditions apply.
}
\end{proposition}

The analytical solution requires a one-dimensional search over a Lagrange multiplier for one reliability constraint, or at most a two-dimensional search for both. Under the conditions in Appendix~\ref{app:physical-design}, tightening the headway-variance budget lets the same analytical construction meet the original budgets after cross-lap corrections, with an explicit bound on its excess holding cost. No separate optimization with cross-lap corrections is required. The noise floors---the minimum variance levels imposed by new disturbances---are \(R_{\varepsilon}=1\) and \(R_{\delta}=2\); the corresponding designs are obtained by limits, with the single-constraint endpoints summarized in Table~\ref{tab:loop-endpoints}. If the kernel uses only deviations from buses \(n,n-1,\ldots,n-m\), with \(m<N-1\), the modal values are coupled by the reduced coefficient vector, and a compact closed-form optimum is generally unavailable.

\begin{table}[tbp]
\caption{Limiting regimes of the single-constraint optimal designs. The two loose endpoints have the same nonzero-mode multipliers and holding cost.\label{tab:loop-endpoints}}
\centering
\footnotesize\setlength{\tabcolsep}{3pt}
\begin{tabular*}{\textwidth}{@{\extracolsep{\fill}}llccc@{}}
\toprule
Design / endpoint & Limiting operator (zero mode $\lambda_0$) & $J_{\varepsilon,\infty}^{\ell}$ & $J_{\delta,\infty}^{\ell}$ & $J_{D,\infty}^{\ell}$ \\
\midrule
Schedule-adhering, tight ($\eta_\varepsilon\!\to\!\infty$) & $F^*\!\equiv\!0$, all $f_p^*\!\to\!0$ (schedule reset); $\lambda_0\!\to\!0$ & $1$ & $2$ & $1+2\beta+2\beta^2$ \\
Schedule-adhering, loose ($\eta_\varepsilon\!\downarrow\!0$) & $\lambda_q\!\to\!\dfrac{1}{(1+\beta)-\beta z_q^{-1}}$; $\lambda_0\!\to\!1$ & $\infty$ & $2+\dfrac{N-1}{N\beta(1+\beta)}$ & $2\beta(1+\beta)$ \\
Headway-equalizing, tight ($\eta_\delta\!\to\!\infty$) & $f_p^*\!=\!1/N$ (fleet averaging); $\lambda_0\!=\!1$ & $\infty$ & $2$ & $1+2\beta+2\beta^2-\dfrac1N$ \\
Headway-equalizing, loose ($\eta_\delta\!\downarrow\!0$) & $\lambda_q\!\to\!\dfrac{1}{(1+\beta)-\beta z_q^{-1}}$; $\lambda_0\!=\!1$ & $\infty$ & $2+\dfrac{N-1}{N\beta(1+\beta)}$ & $2\beta(1+\beta)$ \\
\bottomrule
\end{tabular*}

\end{table}

The following two subsections investigate two common special cases of holding optimization: one constrained by schedule adherence alone and the other by headway regularity alone.

\subsubsection{Schedule-Adhering Optimal Control Problem}\label{subsubsec:loop-schedule-adhering}

With only a schedule-reliability constraint, \(\eta_{\delta}=0\), and every mode receives the same reliability weight \(\eta_{\varepsilon}\). Differences between optimal modal update factors therefore arise only from the demand-propagation factor \(g_q=1+\beta(1-z_q)\). When \(\beta\ll1\), its magnitude and phase vary little across frequencies, which explains why SC, with the common modal factor \(\lambda_q=f_0\), can be near optimal for schedule-adhering design. For any fixed finite $R_\varepsilon>1$, all optimal modal factors converge to $\sqrt{1-1/R_\varepsilon}$ as $\beta\to0$, so the optimal kernel converges to SC. For every finite budget \(R_{\varepsilon}>1\), \(\eta_{\varepsilon}>0\) and \(0<\lambda_0^*<1\): the controller damps a common displacement of the fleet from its virtual timetable.

Table~\ref{tab:loop-endpoints} summarizes the two endpoints and their performance metrics; their derivation is included in Appendix~\ref{app:loop-optimal-design}. As the budget tightens, the controller approaches the schedule reset \(F_f^*\equiv0\), removing existing deviations before new disturbances enter, but at a relatively high holding cost. As the budget loosens, it approaches the limiting kernel in Table~\ref{tab:loop-endpoints} while \(\lambda_0^*\to1\). This loose endpoint is a stability-boundary limit: schedule variance diverges as timetable anchoring disappears, although every finite-budget design remains strictly stable.

\subsubsection{Headway-Equalizing Optimal Control Problem}\label{subsubsec:loop-headway-equalizing}

For headway-equalizing design, we impose only a headway-variance budget and allow $\lambda_0\in(-1,1]$, while all relative modes remain stable. Proposition~\ref{prop:loop-optimal-design} gives $\lambda_0^*=1$, equivalently $\sum_{p=0}^{N-1}f_p^*=1$, as an optimality result: the controller leaves the common schedule offset unanchored. At this optimum, $h_0=\mu_0=0$, so the common mode drops out of the headway and holding metrics defined above. Optimizing the remaining modes gives the exact frontier for these metrics over the full cyclic kernels in this domain. Appendix~\ref{app:physical-neutral} gives the finite common-mode corrections, while Appendix~\ref{app:physical-design} establishes conditions for retaining feasibility and bounding the design's excess holding cost even when schedule variance diverges.

Proposition~\ref{prop:loop-optimal-design} applies to the nonzero modes with reliability weight \(\eta_{\delta}h_q\). Since \(h_q=4\sin^2(\theta_q/2)\), the modes nearest \(\theta_q=\pi\), representing short--long headway alternation, receive the greatest weight and are damped most strongly. The optimal factors retain the phase of \(g_q\); the supporting calculations are included in Appendix~\ref{app:loop-optimal-design}.

The tight endpoint applies fleet averaging, \(f_p^*=1/N\), removing relative deviations while preserving the fleet's common offset from schedule. The loose endpoint is the same limiting kernel given in Table~\ref{tab:loop-endpoints}. Between these endpoints, the optimal controller trades holding effort against headway regularity while leaving the fleet-average schedule mode neutral.

\section{Discussion: Finite-Horizon Boundary Effects}\label{sec:finite-boundary-discussion}

The stability differences between open chains and loops need not appear immediately in finite operations. This section identifies when their variance responses coincide and how boundary effects develop as disturbances propagate through the fleet. Let $T$ denote the operating horizon, measured in control-point passages per bus: the number of controlled stops traversed on an open corridor, or $T=KS$ for $K$ circuits of a loop with $S$ control points per circuit.

The next proposition links open-chain performance to the loop response before boundary truncation affects a given bus. It identifies a window of equal variances and a lower bound on the eventual open-chain variance (proof in Appendix~\ref{app:transient-envelope}).

\begin{proposition}[Loop condition as open-chain transient envelope]\label{prop:transient-envelope}
Consider the same kernel $F(z)=\sum_{p=-r}^{m}f_pz^p$ under open-chain truncation and periodic indexing on $N$ buses, where $r,m$ are nonnegative integers with $r+m\ge1$. Set $M:=\max_q|F(z_q)|$ and, for each bus $n=1,\ldots,N$, $s_n^\dagger:=1+\min\{\lfloor(n-1)/m\rfloor,\lfloor(N-n)/r\rfloor\}$, omitting a term when its denominator is zero. Under zero initial deviations and the disturbance assumptions of Proposition~\ref{prop:loop-performance-metrics},
\begin{equation}
\Var(\varepsilon_{n,s})=\sigma^2J^{\ell}_{\varepsilon,s}
\qquad\text{for all }1\le s\le s_n^\dagger,
\label{eq:coincidence-window}
\end{equation}
where the left-hand side is the open-chain variance of bus $n$, and $J^{\ell}_{\varepsilon,s}$ is the normalized loop fleet-average variance in \eqref{eq:loop-schedule-finite}, also the normalized variance of every bus by circulant symmetry. If the open-chain transition matrix is Schur stable and $M>1$, then
\begin{equation}
\lim_{s\to\infty}\Var(\varepsilon_{n,s})
\ge\frac{\sigma^2}{N}\frac{M^{2s_n^\dagger}-1}{M^2-1}.
\label{eq:transient-floor}
\end{equation}
\end{proposition}

This result is intuitive: over $j$ transitions, contributing bus displacements lie between $-rj$ and $mj$. With zero initial deviations, the variance at passage $s$ involves at most $s-1$ transitions. Within $s_n^\dagger$, no contributing path leaves either open boundary, so both finite systems reproduce the bi-infinite disturbance response. For leader-only open chains ($r=0$), the last bus $n=N$ has the longest window and hence the strongest variance lower bound.

A stable open chain can therefore accumulate substantial variance before boundary truncation takes effect. For a fixed kernel with $\sup_{|z|=1}|F(z)|>1$, the lower bound \eqref{eq:transient-floor} grows geometrically with $N$ whenever $s_n^\dagger$ grows proportionally with $N$ and the open chain remains stable. Under leader-only feedback, this growth occurs when the bus's distance from the head boundary grows proportionally with $N$; under two-sided feedback, both boundary distances must grow proportionally with $N$. See Section~\ref{subsec:num-boundary} for a numerical illustration.

On the loop, every bus shares the bi-infinite schedule variance through passage $T$ whenever $(r+m)(T-1)<N$, since distinct contributing displacements remain distinct modulo $N$. Together with Proposition~\ref{prop:transient-envelope}, this yields a common local limit at fixed kernel and horizon as $N$ grows: every loop bus and a fraction of open-chain buses tending to one share the bi-infinite variance over that horizon. As $(r+m)T/N$ increases, boundary truncation and periodic overlap can produce different variance responses. Support width determines propagation range; the coefficients determine response magnitudes.

SC is the special case $r=m=0$: schedule-deviation dynamics do not propagate between buses, so the same-lap loop and open-chain schedule and headway formulas agree with those of the bi-infinite system. In steady state, loop holding variance, including cross-lap corrections, also equals that of interior open-chain buses.

\section{Numerical Experiments}\label{sec:numerics}

The main purpose of this section is to illustrate and validate the theory developed in Sections~\ref{sec:modelclasses}--\ref{sec:finite-boundary-discussion}, consistent with the paper's analytical focus. We first examine stability by comparing open-chain and loop responses under a common kernel (Section~\ref{subsec:num-boundary}) and quantifying the conservatism of the coefficient-sum test of \citet{xuan2011} for longer kernels (Section~\ref{subsec:num-conservatism}). We then compare named control laws with analytical and numerical Pareto frontiers and examine sensitivity to demand, fleet size, and reliability budgets (Section~\ref{subsec:num-frontiers}), before assessing the three-tap approximation under same-lap input prediction errors (Section~\ref{subsec:num-robustness}).

Unless otherwise specified, the fleet size is $N=12$ and the demand parameter is $\beta=0.05$; all variance metrics are normalized by $\sigma^2$. Performance comparisons in Sections~\ref{subsec:num-frontiers} and~\ref{subsec:num-robustness} include cross-lap corrections, with $S=12$ unless otherwise specified; the other loop performance comparisons use the metrics of Section~\ref{subsec:loop-performance-metrics}. Curves are obtained from analytical expressions and exact covariance recursions, with the stability regions in Figure~\ref{fig:num-shadow}(b) and full cyclic designs computed numerically, while three-tap designs use the analytical construction and a scalar search. Selected data points are obtained from 1{,}000 independent Monte Carlo replications of \eqref{eq:local-closed} under the corresponding boundary condition (30{,}000 for Figure~\ref{fig:num-frontiers}(b)).

\subsection{Stability Comparisons}\label{subsec:num-stability}

The stability regions in this subsection are independent of $\beta$ when the kernel coefficients are held fixed, because demand compensation removes this parameter from \eqref{eq:local-closed}.

\subsubsection{Open-Chain and Loop Responses under a Common Kernel}\label{subsec:num-boundary}

Figure~\ref{fig:num-boundary}(a) shows the stability regions of the
two-tap kernel $F(z)=f_0+f_1z$, the smallest family for which open-chain and loop stability can differ. The open chain uses current- and leading-bus information; the loop uses exact same-lap inputs at the seam.

\begin{figure}[hbp]
\centering
\includegraphics[width=\textwidth]{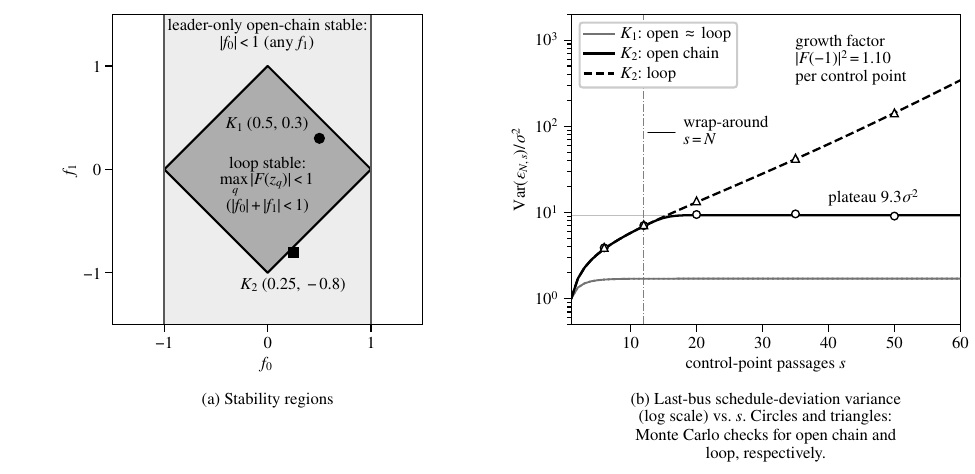}
\caption{Open-chain and loop responses under $F(z)=f_0+f_1z$ ($\beta=0.05$, $N=12$).\label{fig:num-boundary}}

\end{figure}

The open-chain region is the vertical strip
$|f_0|<1$, independent of $f_1$ (Proposition~\ref{prop:open-spectrum}); the
loop region is $\max_q|f_0+f_1z_q|<1$ (Proposition~\ref{prop:loop-stability}).
For a real two-tap kernel the frequency response peaks at $z=1$ or
$z=-1$, both included among the loop's discrete points $z_q$ when $N$ is even, so
the loop region here is exactly the diamond $|f_0|+|f_1|<1$: on this family
the finite loop condition, the bi-infinite condition, and the
coefficient-sum bound of \citet{xuan2011} all coincide. This tightness is
special to two-tap kernels; Section~\ref{subsec:num-conservatism}
quantifies how quickly it fails for longer kernels. For reference, $K_1=(0.5,0.3)$ is stable under both boundaries, and its variance paths in panel~(b) are virtually indistinguishable. Kernel $K_2=(0.25,-0.8)$
lies in the gap: open-chain stable, loop unstable, with the unstable loop
mode at $\theta_q=\pi$---the adjacent-contrast (bunching) pattern. For this current-and-leading-bus kernel family, the gap between the diamond and the strip quantifies the larger asymptotic stability region of the open chain. However, asymptotic stability does not ensure small transient variance. For a fixed kernel in this gap, applying Proposition~\ref{prop:transient-envelope} to the last bus gives a lower bound on its open-chain variance before boundary effects arrive; this bound grows geometrically with fleet size, and Panel~(b) illustrates the transient variance buildup.

Panel~(b) tracks the schedule-deviation variance of the last bus under both
boundaries, computed by iterating the exact covariance recursion
$\boldsymbol\Sigma_{s+1}=\mathbf A\boldsymbol\Sigma_s\mathbf A^{\!\top}+\sigma^2\mathbf I$
with $\mathbf A=f_0\mathbf I+f_1\mathbf S$ for the two shifts
$\mathbf S=\mathbf L$ and $\mathbf S=\mathbf P$. Expanding
$\mathbf A^{s}=\sum_{k=0}^{s}\binom{s}{k}f_0^{\,s-k}f_1^{\,k}\mathbf S^{k}$,
the two systems share every term with $k<N$, because $\mathbf L^{k}$ and
$\mathbf P^{k}$ act identically on the last bus there; the first term at
which they can differ is $k=N$, where $\mathbf L^{N}=\mathbf 0$ while
$\mathbf P^{N}=\mathbf I$. That term can appear only once $s\ge N$.
The resulting variance difference first appears at $s=N+1=13$ and becomes visible by $s\approx17$ and decisive after
$s\approx25$---confirming the delayed influence of boundaries described in
Section~\ref{sec:finite-boundary-discussion}. The open chain, with every
eigenvalue equal to $f_0=0.25$, settles at a finite plateau of
$9.3\sigma^2$, consistent with the steady-state bound of
Corollary~\ref{cor:open-output-bounds}. For reference, the same $f_0$ with
$f_1=0$ would settle at $\sigma^2/(1-f_0^2)\approx1.07\sigma^2$. The large variance under $K_2$ illustrates why open-chain stability alone does not ensure good performance. This kernel is used here to diagnose boundary effects, rather than as a proposed operational design. The loop instead diverges geometrically at rate
$|F(-1)|^2\approx1.10$ per control point in the $\theta_q=\pi$ mode, i.e., as a
growing early--late alternation between adjacent buses.
Although the loop here uses same-lap feedback, Appendix~\ref{app:seam-delay} shows that it remains unstable under $K_2$ even with cross-lap feedback.

The same control kernel is therefore stable on the open chain but unstable on the loop, with the last-bus variance paths separating only after disturbances have had time to propagate around the fleet. While the examples above use leader-only kernels on the open chain, recall that the two-sided example in Section~\ref{subsec:loop-stability-xuan} shows that boundary effects persist when both systems also use following-bus feedback.

\subsubsection{Conservatism of the Coefficient-Sum Test}\label{subsec:num-conservatism}

The coefficient-sum test $\sum_p|f_p|<1$ of \citet{xuan2011} is sufficient for loop stability but can exclude stable controllers. It coincides with the exact stability condition for two-tap kernels, as shown in Figure~\ref{fig:num-boundary}(a). Figure~\ref{fig:num-shadow} shows how the gap grows as more feedback coefficients are allowed.

To quantify this conservatism, let $r_m$ be the maximum of $\sum_{p=0}^m|f_p|/\sup_\theta|F(e^{-\mathrm{i}\theta})|$ over nonzero kernels supported on $0,\ldots,m$. The denominator is the largest gain over all frequencies and bounds the finite loop modal maximum; $r_m$ measures the largest overestimate from the coefficient sum. We have $r_1=1$; for three-tap kernels, $r_2=\sqrt2$, attained by $f\propto(1,-2,-1)$. The longer kernels $f=(4,6,-3,2)$ and $f=(1,2,2,-2,1)$ give the exact lower bounds $r_3\ge5/3\approx1.67$ and $r_4\ge2$, respectively (derivations omitted), while Cauchy--Schwarz gives the general upper bound $r_m\le\sqrt{m+1}$.

Figure~\ref{fig:num-shadow}(a) uses $f=[0.90/(2\sqrt2)](1,-2,-1)\approx(0.318,-0.636,-0.318)$, marked by a star in the figure. This kernel retains the worst-case ratio ($\sqrt2$) among three-tap kernels. The coefficient-sum test rejects this kernel ($\sum_p|f_p|=0.9\sqrt2\approx1.27$), yet $\max_q|F(z_q)|=0.90$, giving a $10\%$ loop stability margin.

\begin{figure}[tp]
\centering
\includegraphics[width=\textwidth]{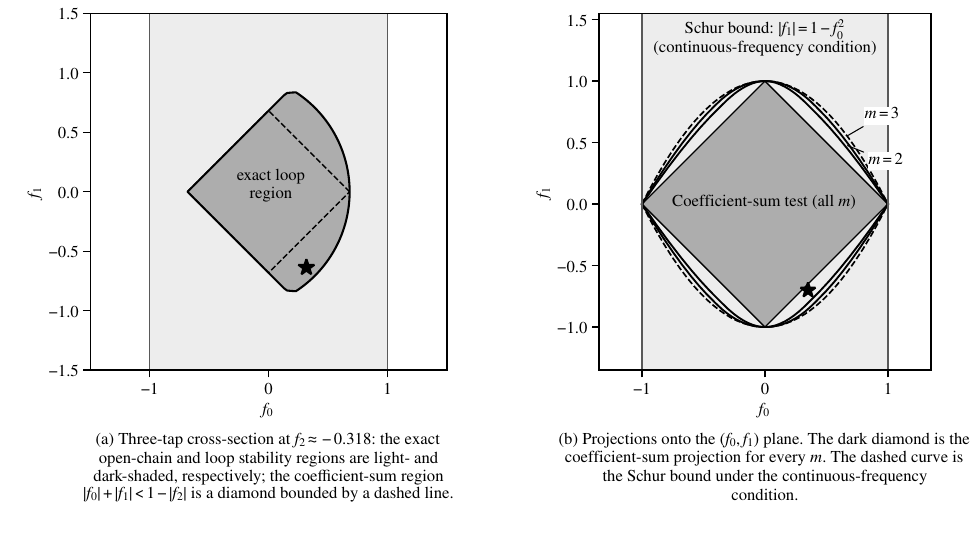}
\caption{Conservatism of the coefficient-sum test ($N=12$).\label{fig:num-shadow}}

\end{figure}

Figure~\ref{fig:num-shadow}(b) shows the admissible pairs $(f_0,f_1)$: each pair permits at least one choice of the remaining coefficients $f_2,\ldots,f_m$ that makes the loop stable. The coefficient-sum test gives the same diamond $|f_0|+|f_1|<1$ for every $m$, whereas the exact projection expands with $m$. At $f_0=0.5$, the upper limit on $|f_1|$ rises from $0.50$ ($m=1$) to $0.65$ ($m=2$) and $0.70$ ($m=3$); see the two solid curves outside the diamond. The starred pair $(0.35,-0.7)$ becomes stable when $f_2=-0.35$ is added, reducing the maximum modal gain from $1.05$ to $0.99$. Thus, additional feedback coefficients can stabilize combinations excluded by the coefficient-sum test.

The dashed curve is the Schur bound $|f_1|\le1-f_0^2$ under the stronger condition $|F(e^{-\mathrm{i}\theta})|\le1$ for every $\theta$ \citep{caratheodory1911,schur1917}. This bounds $|f_1|$ for all kernels satisfying that condition, regardless of kernel length.

\subsection{Comparison of Optimal Frontiers and Existing Control Schemes}\label{subsec:num-frontiers}

Figure~\ref{fig:num-frontiers} compares (i) best-found full cyclic designs obtained numerically with cross-lap corrections (see Appendices~\ref{app:physical-metrics}--\ref{app:physical-bounds}); (ii) the analytical design of Proposition~\ref{prop:loop-optimal-design}; (iii) Pareto-efficient designs of the named control kernels SC, FHC, BHC, and TWHC under the same-lap convention; and (iv) a near-optimal three-tap design. Panel~(a) imposes a schedule-variance budget, while panel~(b) imposes a headway-variance budget. Full analytical designs use $\times$ markers; the three-tap approximation uses squares.

\begin{figure}[htbp]
\centering
\includegraphics[width=\textwidth]{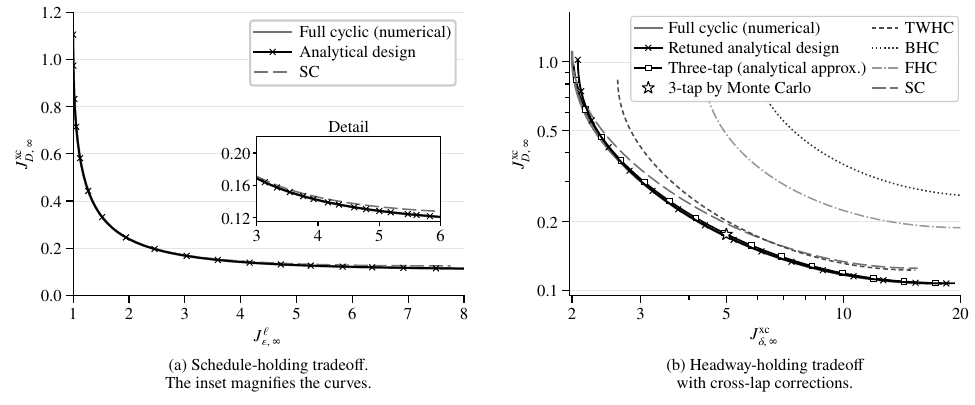}
\caption{Pareto designs and named control laws on a loop ($\beta=0.05$, $N=S=12$; variances normalized by $\sigma^2$).\label{fig:num-frontiers}}

\end{figure}

We use the superscript $\mathrm{xc}$ for outputs and metrics with cross-lap corrections, as defined in Appendix~\ref{app:physical-metrics}. Panel~(a) minimizes $J_{D,\infty}^{\mathrm{xc}}$ under $J_{\varepsilon,\infty}^{\ell}\le R_\varepsilon$, with all modes strictly stable. The full cyclic design uses a numerical search over the Lagrange multiplier, with each modal value optimized from multiple starting points. For both panels, denser multistart checks of the complex modal solves yielded objective value spreads below $10^{-12}$. The analytical comparator retains Proposition~\ref{prop:loop-optimal-design}'s closed-form construction with $\eta_\delta=0$ and its one-dimensional search over $\eta_\varepsilon$. The holding cost is reevaluated using cross-lap corrections. SC (long dashes) is optimized within its one-parameter family under the same budget, while pure headway rules are absent because their schedule variance diverges (Corollary~\ref{cor:headway-based-schedule-divergence}).

For schedule budgets from $1.05$ to $8$, the analytical cost is within $0.10\%$ of the full cyclic numerical cost. SC's excess cost is $0.26\%$ at $R_\varepsilon=1.5$ and $4.00\%$ at $R_\varepsilon=5$, consistent with the near-optimality observed numerically by \citet{xuan2011} and the convergence of the schedule-adhering optimum to SC as \(\beta\to0\) (Section~\ref{subsubsec:loop-schedule-adhering}). However, the small gap of SC grows with demand: at budget $R_\varepsilon=5$, SC's excess cost increases to $10.76\%$ at $\beta=0.10$, while the analytical design remains within \(0.18\%\) of the full cyclic numerical cost. This reveals the advantage of the analytical schedule-adhering design for high-demand bus routes.

Panel~(b) compares three design families using the exact headway and holding metrics of Appendices~\ref{app:physical-metrics}--\ref{app:physical-neutral}. The full cyclic design uses the same numerical procedure as in panel~(a), now under a headway-variance budget, allowing $\lambda_0\in(-1,1]$ while keeping the nonzero modes strictly stable. The named headway-based laws retain their neutral common mode. The retuned analytical design preserves the closed-form modal construction of Proposition~\ref{prop:loop-optimal-design}, with $\eta_\varepsilon=0$ and $\lambda_0=1$. It is retuned to meet the exact headway-variance budget with cross-lap corrections via a bisection search of $\eta_\delta$. For example, at a headway-variance budget $R_\delta=5$, the unretuned design satisfies $J_{\delta,\infty}^{\ell}=5$ but gives $J_{\delta,\infty}^{\mathrm{xc}}=5.13710$; retuning changes $\eta_\delta$ from $0.02645$ to $0.02857$, giving $J_{\delta,\infty}^{\mathrm{xc}}=5$. Both plotted coordinates include the corrections.

For headway-variance budgets from $3$ to $16$, the retuned analytical design is only $0.11$--$0.35\%$ above the full cyclic numerical cost. The difference rises to $16.5\%$ at budget $2.1$: the analytical family's tight endpoint is $J_{\delta,\infty}^{\mathrm{xc}}=2+(S-2)/N^2=2.0694$, whereas allowing the common-mode multiplier $\lambda_0$ to vary, rather than fixing it at $1$, permits tighter budgets. Varying the circuit length to $S=1,12,24$ at budget $5$ gives cost differences of $0.011\%$, $0.142\%$, and $0.275\%$, respectively. These small gaps provide numerical evidence that the analytical design, with scalar retuning, remains an effective approximation to the full cyclic numerical frontier under cross-lap corrections across moderate or loose budgets.

The three-tap approximation retains only the current-bus and immediate-neighbor coefficients $f_{-1}^*(0,\eta_\delta)$, $f_0^*(0,\eta_\delta)$, and $f_1^*(0,\eta_\delta)$ of Proposition~\ref{prop:loop-optimal-design}'s analytical kernel, setting all other coefficients to zero. As with the retuned analytical design, it requires only closed-form evaluation and a one-dimensional search over $\eta_\delta$ to meet the corrected headway-variance budget, with stability checked after truncation. Across the sampled budgets from $3$ to $16$, its holding cost is within $2.7\%$ of the full cyclic numerical frontier and below each feasible named law; the gap rises to $8.7\%$ at budget $2.05$. All plotted three-tap designs are schedule-stable.

Among the named laws, SC outperforms BHC and FHC and is surpassed by TWHC (short dashes) only at loose budgets, consistent with previous findings \citep{xuan2011,gu2025}. Nevertheless, what is new here is that SC's gap to the full cyclic numerical frontier remains substantial except near the tight-budget limit. At $R_\delta=5$, its holding cost is $14.0\%$ above that frontier and $11.3\%$ above the three-tap approximation. TWHC's deterioration as budgets tighten reflects the larger feedback gains $\alpha$ needed to meet the headway constraint (Appendix~\ref{app:twhc}); BHC incurs a further holding-cost penalty from its demand-dependent damping requirement (Appendix~\ref{app:bhc}).

The benefits of the full cyclic and three-tap designs persist over wider parameter ranges. We vary one parameter at a time from $(R_\delta,\beta,N)=(5,0.05,12)$, sampling $R_\delta=3$--$16$, $\beta=0.02$--$0.10$, and $N=8$--$24$, with $S=12$. Using cross-lap-corrected metrics, SC's holding cost is $11.7$--$17.3\%$ above the full cyclic numerical frontier, whereas the retuned analytical and three-tap designs remain within $0.35\%$ and $3.5\%$, respectively.

Together, these comparisons show that both the retuned full analytical kernel and its three-tap approximation outperform all four established controls considered here at the tested moderate and loose headway-variance budgets. This demonstrates the model's value for both benchmarking existing controls and deriving improved designs through closed-form expressions and scalar tuning.

\subsection{Sensitivity to Independent Errors in Same-Lap Inputs}\label{subsec:num-robustness}

This subsection provides a supplementary sensitivity analysis of the loop controls when their required future same-lap inputs are perturbed by independent additive errors. We focus on the three-tap approximation studied in Section~\ref{subsec:num-frontiers} and evaluate headway and holding variances with the cross-lap corrections of Appendix~\ref{app:physical-metrics}. Synthetic predictions have also been used to examine the sensitivity of bus holding policies to prediction errors \citep{berrebi2018a}.

The three-tap kernel requires two types of prediction. Buses $1,\ldots,N-1$ use their following bus's future deviation, with error standard deviation denoted by $\rho_F\cdot\sigma$, while bus~1 also uses bus~$N$'s future same-lap deviation in the leading-bus kernel term, with error standard deviation $\rho_L\cdot\sigma$ (typically $\rho_L>\rho_F$). Bus~$N$'s same-lap following-bus input is the already observed $\varepsilon_{1,s}$, rather than its physical follower's next-lap deviation $\varepsilon_{1,s+S}$, so it requires no prediction. We represent errors in these unobserved inputs by prescribed zero-mean additive perturbations. The errors are independent across buses, control points, and prediction types, and independent of the current and past state vectors and cruising disturbances. Their covariance contribution is

\begingroup

\begin{equation}
\boldsymbol\Omega=\sigma^2\diag(\tau_1,\ldots,\tau_N),\qquad
\tau_n=\begin{cases}
f_{-1}^2\rho_F^2+f_1^2\rho_L^2, & n=1,\\
f_{-1}^2\rho_F^2, & n=2,\ldots,N-1,\\
0, & n=N.
\end{cases}
\label{eq:input-error-intensity}
\end{equation}
\endgroup
In particular, SC has $\boldsymbol\Omega=\mathbf0$. Since $\boldsymbol\Omega$ is generally noncirculant, the cross-lap-corrected output variances cannot in general be obtained by multiplying their prediction-error-free values by a common factor.

For a strictly stable kernel, let $\boldsymbol\Sigma$ be the stationary state covariance under cruising disturbances and prediction errors and define $\boldsymbol\Gamma$ as the current--preceding-lap covariance:
{
\begin{equation}
\boldsymbol\Sigma
=\mathbf A_\ell\boldsymbol\Sigma\mathbf A_\ell^\top
+\sigma^2\mathbf I+\boldsymbol\Omega,
\qquad
\boldsymbol\Gamma=\mathbf A_\ell^S\boldsymbol\Sigma.
\label{eq:input-error-covariance}
\end{equation}
}
Recall that $\mathbf e_n\in\mathbb R^N$ is the $n$th standard basis vector. With $\mathbf G_\ell=\mathbf I-\mathbf P$ and $\mathbf M_\ell=\mathbf A_\ell-(1+\beta)\mathbf I+\beta\mathbf P$, the corrected steady-state metrics are

\begingroup

\begin{equation}
\begin{aligned}
\widetilde J_{\delta,\infty}^{\mathrm{xc}}
&=\frac{1}{N\sigma^2}\Bigl[
\tr(\mathbf G_\ell\boldsymbol\Sigma\mathbf G_\ell^\top)
+2\mathbf e_1^\top(\boldsymbol\Sigma-\boldsymbol\Gamma)\mathbf e_N
\Bigr],\\[2pt]
\widetilde J_{D,\infty}^{\mathrm{xc}}
&=\frac{1}{N\sigma^2}\Bigl[
\tr(\mathbf M_\ell\boldsymbol\Sigma\mathbf M_\ell^\top)+\tr(\boldsymbol\Omega)
+2\beta\mathbf e_1^\top\mathbf M_\ell(\boldsymbol\Gamma-\boldsymbol\Sigma)\mathbf e_N
+2\beta^2\mathbf e_N^\top(\boldsymbol\Sigma-\boldsymbol\Gamma)\mathbf e_N
\Bigr].
\end{aligned}
\label{eq:input-error-metrics}
\end{equation}
\endgroup
Under this exogenous additive-noise model, finite error variances change the noise covariance but not $\mathbf A_\ell$, so they do not create a new deterministic stability loss for a fixed strictly stable kernel.

The parameters $\rho_F$ and $\rho_L$ measure residual forecast error, not the raw accumulation of cruising disturbances. If $k$ independent future control-point disturbances, each with variance $\sigma^2$, simply accumulated without intervening control or predictive information, their standard deviation would be $\sqrt{k}\sigma$. This gives only a \textit{loose} no-prediction benchmark for the error scale factors, $\sqrt{k}$, if the predicted bus is $k$ control points behind the bus to be held. In practice, future holding actions at those control points can attenuate accumulated deviations of the predicted bus, and real-time travel-time prediction can exploit traffic conditions, current bus states and speeds, signal timing, and historical time-of-day patterns. Both effects may reduce $\rho_F$ and $\rho_L$ below this benchmark. In addition, the following-bus prediction horizon on a loop is about $S/N$ control points on average. When $S<N$ (e.g., when only a handful of stops are control points), the short horizon may allow even the no-prediction benchmark to be less than 1. Although the seam horizon is much longer, its error enters only bus~1's holding decision and may have a smaller aggregate effect than follower-prediction errors. Accordingly, we use $(\rho_F,\rho_L)=(0.5,1)$ and $(0.5,2)$ as moderate illustrative cases and $(1,3)$ as a deliberately severe stress case in the following tests, without claiming empirical calibration.

We use headway regularity as the sole reliability constraint. Table~\ref{tab:prediction-error-sensitivity} reports two evaluations. First, each prediction-error-free three-tap design is applied without changing its kernel (unchanged design), and the percentage increases in headway and holding variances, $\mathcal E_\delta$ and $\mathcal E_D$, measure the performance changes caused by prediction errors. Second, under each prediction-error specification, only the one-dimensional search over $\eta_\delta$ is repeated to restore the same corrected headway-variance budget (retuned design); the resulting holding-cost saving is reported relative to SC.

\begin{table}[htbp]
\caption{Sensitivity of the three-tap design to same-lap prediction errors.\label{tab:prediction-error-sensitivity}}
\centering
{\footnotesize
\setlength{\tabcolsep}{4pt}
\begin{tabular*}{\textwidth}{@{\extracolsep{\fill}}lccc@{}}
\toprule
$(\rho_F,\rho_L)$ & $\mathcal E_\delta$ (\%) & $\mathcal E_D$ (\%) & Retuned saving vs. SC (\%)\\
\midrule
$(0.5,1)$ & $0.40/0.16/0.09$ & $1.69/1.08/0.74$ & $7.06/9.02/9.89$\\
$(0.5,2)$ & $0.64/0.22/0.10$ & $2.58/1.39/0.81$ & $5.90/8.70/9.82$\\
$(1,3)$ & $2.01/0.75/0.36$ & $8.24/4.84/3.07$ & $-1.56/5.14/7.68$\\
\bottomrule
\end{tabular*}}
\par\smallskip
{\footnotesize Notes. $N=S=12$, $\beta=0.05$; entries are ordered by $R_\delta=3/5/8$. $\mathcal E_\delta$ and $\mathcal E_D$ refer to the unchanged design; savings compare the retuned design with SC at the same corrected headway-variance budget.\par}

\end{table}

Sensitivity increases as the headway-variance budget tightens. For the two moderate cases, even at $R_\delta=3$ the increases in headway and holding variances under the unchanged design are at most $0.64\%$ and $2.58\%$, respectively. The corresponding fleet-level headway standard deviation, which scales with the square root of the headway-variance metric, increases by at most $0.32\%$. At the default $R_\delta=5$, even the stress case $(1,3)$ increases headway variance by only $0.75\%$ (standard deviation by $0.37\%$) and holding variance by $4.84\%$. Under that stress case and the tighter $R_\delta=3$, however, the variance increases reach $2.01\%$ and $8.24\%$, and one-dimensional retuning leaves the three-tap cost $1.56\%$ above SC. Thus moderate prediction errors cause only small performance changes in the unchanged designs and preserve the retuned three-tap advantage over SC, whereas sufficiently large errors combined with a tight reliability budget can materially erode or eliminate that advantage.

Analysis of prediction errors in full cyclic kernels is beyond the present scope. The near-optimal three-tap performance in Section~\ref{subsec:num-frontiers} shows that the current-bus and nearest-neighbor taps capture most of the prediction-error-free performance gain. Where the remaining coefficients are small, their prediction errors may have a limited aggregate effect. Full cyclic designs also start with a larger prediction-error-free holding-cost margin over SC than three-tap designs, which may help them retain an advantage under prediction errors.

Also note that the one-dimensional retuning above is a simple feasibility correction, not a full error-aware reoptimization. Direct error-aware kernel optimization can match or improve on the retuned design and may reduce reliance on inputs with larger prediction errors, particularly in severe-error cases. Such an optimization would give up the one-dimensional analytical structure emphasized in this paper and is left for future work.

\section{Conclusions}\label{sec:conclusions}

This paper establishes a common framework for adaptive linear bus holding on finite open chains and loops, relating these settings to the bi-infinite system studied in the literature. Matrix and Fourier methods yield necessary and sufficient stability conditions for leader-only open-chain kernels and general cyclic loop kernels, exact variance formulas on loops with computable covariance bounds on open chains, and closed-form Pareto-optimal designs over the full cyclic family. The exact conditions admit kernels that the coefficient-sum test rejects, and the gap widens with kernel length. Cross-lap output corrections connect same-lap analytical designs to physical headway and holding performance, with conditions for feasibility and bounds on excess holding cost.

The results make boundary conditions and the operating horizon explicit design considerations. With current- and leading-bus feedback, open-chain stability depends only on the current-bus coefficient, whereas the loop boundary brings the entire kernel into the stability condition; stability can therefore differ between the two boundaries in either direction, and a stable open chain can still accumulate substantial variance before boundary truncation takes effect. Thus, asymptotic stability alone does not ensure satisfactory operating performance.

Modal analysis connects control coefficients to the operating mechanisms they act on. It distinguishes fleet-wide schedule drift, which headway-based feedback cannot anchor, from relative patterns such as alternating short--long headways that strongly affect headway regularity. The same modal gains and output weights explain the behavior of existing rules, including BHC and TWHC, and determine how optimal designs allocate control effort. The analysis also shows why the schedule-adhering optimum approaches SC as demand intensity decreases, whereas the headway-equalizing optimum leaves the common schedule mode neutral and damps relative deviations most strongly where headway weights are largest. This brings control coefficients, observable operating patterns, and design choices into one analytical framework.

The analytical frontiers quantify the minimum holding variance (as a proxy for holding cost) for prescribed schedule or headway variance limits. At the tested moderate and loose headway-variance budgets, the retuned analytical and three-tap designs outperform all four established controls considered. Their holding costs remain within $0.35\%$ and $3.5\%$, respectively, of the full cyclic numerical frontier across these budgets and other parameter variations. Using only the current bus and its immediate neighbors, the three-tap construction captures most of the full kernel's benefit with closed-form coefficient expressions and a one-dimensional search. In the supplementary sensitivity tests, scalar retuning preserves holding-variance savings over SC under low to moderate input errors, while the largest-error setting eliminates those savings at the tightest tested budget.

The main limitations are the homogeneous linear operating model and the use of future same-lap information in the loop analysis. The supplementary sensitivity study considers prescribed independent additive input errors; prediction biases, error correlations, and residuals generated by a causal forecasting model remain outside its scope. The holding-variance objective also leaves fleet and timetable decisions outside the design problem.

A natural extension is to allow demand intensity, disturbance variance, and control kernels to vary across control points, retaining the structure of interactions along the bus sequence that supports modal analysis. The analytical framework could then guide the selection of control points and the design of location-specific kernels for heterogeneous systems, while accounting for imperfect compensation due to demand-estimation errors. Further extensions to nonlinear holding laws that truncate negative holds at zero (allowing individual slacks to be optimized rather than set by a three-sigma rule; see \citet{gu2025}) and to heterogeneity across buses would require methods beyond the present analytical framework.

Another priority is richer prediction-error modeling and analysis, followed by error-aware kernel design. Extending the modal, covariance, and augmented-state formulations to correlated errors and, where appropriate, cross-lap feedback could reveal how forecast quality should shape control weights and information requirements.

A further planning extension is to jointly choose fleet size, scheduled slack, and the holding kernel. On a fixed-fleet loop, extra slack lengthens scheduled headways, so tighter headway regulation may reduce variability-related waiting while increasing baseline waiting through longer scheduled headways. Embedding the stability constraints and reliability--holding frontiers in passenger-cost optimization would clarify when stronger control, additional vehicles, or timetable adjustments deliver the greatest service benefit.

\section*{Acknowledgments}
The work described in this paper was partially supported by a grant from the Research Grants Council of the Hong Kong Special Administrative Region, China (PolyU P0051967/RGC 15237624), by the Otto Poon Charitable Foundation through an Interdisciplinary Research Fund in Smart Cities (Project No. P0058095), and by the National Natural Science Foundation of China (Grant Nos. 72201214, 72531008, and 72571220).

\clearpage
\appendix

\renewcommand{\theHsection}{appendix.\Alph{section}}
\renewcommand{\theHsubsection}{\theHsection.\arabic{subsection}}
\renewcommand{\theHsubsubsection}{\theHsubsection.\arabic{subsubsection}}

\section{Open-Chain Analysis: Proofs and Supplementary Details}\label{app:open-chain}

\subsection{Proof of Proposition~\ref{prop:open-cov}}\label{app:open-cov}

By Proposition~\ref{prop:open-spectrum}, $|f_0|<1$ implies $\mathbf{A}_{\mathrm{o}}$ is Schur stable. Iterating the covariance recursion gives
\[
\boldsymbol{\Sigma}_{s}
=\mathbf{A}_{\mathrm{o}}^{s}\boldsymbol{\Sigma}_{0}(\mathbf{A}_{\mathrm{o}}^{\top})^{s}
+\sum_{j=0}^{s-1}\mathbf{A}_{\mathrm{o}}^{j}\mathbf{Q}(\mathbf{A}_{\mathrm{o}}^{\top})^{j}.
\]
The first term vanishes as $s\to\infty$, and the series converges absolutely in any matrix norm because $\mathbf{A}_{\mathrm{o}}$ is Schur stable. The limit satisfies the discrete Lyapunov equation. Uniqueness follows from the standard discrete Lyapunov theorem for Schur-stable matrices.

\subsection{Proof of Corollary~\ref{cor:open-output-bounds}}\label{app:open-output-bounds}

For zero initial deviation, iterating the covariance recursion and using the positive semidefiniteness of each summand gives
\begingroup
\[
\boldsymbol\Sigma_s
=\sum_{j=0}^{s-1}\mathbf A_{\mathrm o}^{j}\mathbf Q(\mathbf A_{\mathrm o}^{\top})^{j}
\preceq\sum_{j=0}^{\infty}\mathbf A_{\mathrm o}^{j}\mathbf Q(\mathbf A_{\mathrm o}^{\top})^{j}
=\boldsymbol\Sigma_\infty.
\]
\endgroup
Here $\preceq$ denotes the Loewner order: $\mathbf{X}\preceq\mathbf{Y}$ means that $\mathbf{Y}-\mathbf{X}$ is positive semidefinite. Congruence preserves this order, so $\mathbf{C}\boldsymbol{\Sigma}_{s}\mathbf{C}^{\top}\preceq\mathbf{C}\boldsymbol{\Sigma}_{\infty}\mathbf{C}^{\top}$ for any matrix $\mathbf{C}$.

Schedule and headway deviations correspond to $\varepsilon_{n,s}=\mathbf{e}_n^{\top}\boldsymbol{\varepsilon}_{s}$ and $\delta_{n,s}=\mathbf{g}_n^{\top}\boldsymbol{\varepsilon}_{s}$, respectively. Holding times satisfy $D_{n,s}=d_s+\mathbf{m}_n^{\top}\boldsymbol{\varepsilon}_{s}$. The limiting schedule- and headway-deviation variances and the square root of the limiting holding-time variance give the equalities in \eqref{eq:open-schedule-var}--\eqref{eq:open-holding-sd}. The finite-horizon comparison follows from the congruence ordering above.

It remains to justify the computable norm bound. For each summand in \eqref{eq:open-sigma-infty},
\[
\mathbf{A}_{\mathrm{o}}^{j}\mathbf{Q}(\mathbf{A}_{\mathrm{o}}^{\top})^{j}
\preceq
\lambda_{\max}\!\left(\mathbf{A}_{\mathrm{o}}^{j}\mathbf{Q}(\mathbf{A}_{\mathrm{o}}^{\top})^{j}\right)\mathbf{I}
\preceq
\lVert\mathbf{A}_{\mathrm{o}}^{j}\rVert_2^2\lVert\mathbf{Q}\rVert_2\mathbf{I}.
\]
Here $\lambda_{\max}(\cdot)$ denotes the largest eigenvalue. The first inequality uses the spectral bound for positive semidefinite matrices, and the second uses submultiplicativity of the spectral norm. Summing over $j$ gives
\[
\boldsymbol{\Sigma}_{\infty}
\preceq
\left(
\sum_{j=0}^{\infty}
\lVert\mathbf{A}_{\mathrm{o}}^{j}\rVert_2^2\lVert\mathbf{Q}\rVert_2
\right)\mathbf{I}.
\]
Finally, by \eqref{eq:open-binomial} and the triangle inequality,
\[
\lVert\mathbf{A}_{\mathrm{o}}^{j}\rVert_2
\le
\sum_{k=0}^{\min\{j,\nu-1\}}
\binom{j}{k}|f_0|^{j-k}\lVert\mathbf{R}\rVert_2^{k}.
\]
Combining these two bounds yields the stated definition of $\mathcal{B}_{\mathrm{o}}$. Applying $\mathbf{a}^{\top}\boldsymbol{\Sigma}_{\infty}\mathbf{a}\le\lVert\mathbf{a}\rVert_2^2\mathcal{B}_{\mathrm{o}}$ with $\mathbf{a}=\mathbf{e}_n,\mathbf{g}_n,\mathbf{m}_n$ and $\lVert\mathbf{g}_n\rVert_2^2=2$ gives the inequalities in \eqref{eq:open-schedule-var}--\eqref{eq:open-headway-var} and, after taking square roots, \eqref{eq:open-holding-sd}.

\subsection{Nilpotency index of an open chain}\label{app:supp-nilpotency-index}\label{app:supplementary-details}

The nilpotency index $\nu$ of $\mathbf R$ can be read directly from the first nonzero shift coefficient. If $\ell_*:=\min\{\ell\in\{1,\ldots,\min(m,N-1)\}: f_\ell\ne0\}$ exists, then $\nu=\lceil N/\ell_*\rceil$, where $\lceil x\rceil$ indicates the smallest integer not less than $x$. If no such $\ell_*$ exists, then $\mathbf R=\mathbf 0$ and we take $\nu=1$. In the common case $f_1\ne0$, this gives $\nu=N$.

\subsection{Proof of Corollary~\ref{cor:open-schedule-headway-norm}}\label{app:open-schedule-headway-norm}

The matrix $\mathbf{B}_{\mathrm{o}}:=\mathbf{G}_{\mathrm{o}}^{\top}\mathbf{G}_{\mathrm{o}}$ has diagonal entries $(1,2,\ldots,2,1)$, adjacent off-diagonal entries $-1$, and all other entries zero. Its eigenvalues are
\[
\lambda_k=2-2\cos\!\left(\frac{k\pi}{N}\right),
\qquad k=0,\ldots,N-1.
\]
Thus $\mathbf{0}\preceq\mathbf{B}_{\mathrm{o}}\preceq[2+2\cos(\pi/N)]\mathbf I$. Since $\boldsymbol{\Sigma}_{s}\succeq0$, cyclic invariance of the trace gives $V_{\delta,s}^{\mathrm{o}}=\tr(\mathbf{B}_{\mathrm{o}}\boldsymbol{\Sigma}_{s})/(N-1)$, and hence
\[
0\le V_{\delta,s}^{\mathrm{o}}
\le\frac{N[2+2\cos(\pi/N)]}{N-1}V_{\varepsilon,s}^{\mathrm{o}}
=\gamma_N V_{\varepsilon,s}^{\mathrm{o}},
\]
which proves \eqref{eq:open-schedule-headway-bound}. The steady-state version follows by taking the limit $s\to\infty$ whenever $\boldsymbol{\Sigma}_{s}$ converges.

\subsection{Proof of Proposition~\ref{prop:simple-point-metrics}}\label{app:simple-point-metrics}

Under SC, \(\mathbf{A}_{\mathrm{o}}=f_0\mathbf{I}\), so with zero initial deviation the SC recursion gives $\varepsilon_{n,s}=\sum_{k=0}^{s-1}f_0^k v_{n,s-k}$, rendering $\varepsilon_{n,s}$ and $\varepsilon_{n',s}$ independent if $n\ne n'$.
Independence over control points and \(\Var(v_{n,s})=\sigma^2\) imply
\[
\Var(\varepsilon_{n,s})
=
\sigma^2\sum_{k=0}^{s-1}f_0^{2k}
=
\sigma^2\frac{1-f_0^{2s}}{1-f_0^2}.
\]
This proves \eqref{eq:simple-Jepsilon}.

Under the independent homoskedastic disturbance assumption, \(\boldsymbol{\Sigma}_{s}\) is a scalar multiple of \(\mathbf I\): $\boldsymbol{\Sigma}_{s}=\sigma^2 J_{\varepsilon,s}(f_0)\mathbf I$. For each headway, $\delta_{n,s}=\varepsilon_{n,s}-\varepsilon_{n-1,s}$, so independence across buses gives $\Var(\delta_{n,s})=2\Var(\varepsilon_{n,s})$. Summing these $N-1$ variances and dividing by $(N-1)\sigma^2$ yields $J_{\delta,s}(f_0)=2J_{\varepsilon,s}(f_0)$, proving \eqref{eq:simple-Jdelta}.

Finally, substituting \(\mathbf{A}_{\mathrm{o}}=f_0\mathbf{I}\) into the holding output matrix gives, for an interior bus \(n\ge2\),
\[
D_{n,s}-d_s
=
-(1+\beta-f_0)\varepsilon_{n,s}
+
\beta\varepsilon_{n-1,s}.
\]
The deterministic slack \(d_s\) does not affect variance. Independence across buses gives
\[
\Var(D_{n,s})
=
\left[(1+\beta-f_0)^2+\beta^2\right]\Var(\varepsilon_{n,s}).
\]
Combining this holding-variance expression with the schedule variance above gives \eqref{eq:simple-JD}.

\subsection{Proof of Corollary~\ref{cor:simple-asymptotic-opt}}\label{app:simple-asymptotic-opt}

Taking \(s\to\infty\) in Proposition~\ref{prop:simple-point-metrics} gives
\begingroup
\[
J_{\varepsilon,\infty}(f_0)=\frac{1}{1-f_0^2},\quad
J_{\delta,\infty}(f_0)=\frac{2}{1-f_0^2},\quad
J_{D,\infty}(f_0;\beta)=\frac{(1+\beta-f_0)^2+\beta^2}{1-f_0^2}.
\]
\endgroup
The first two quantities are strictly increasing on $[0,1)$, so the schedule and headway constraints imply $f_0\le f_{\max,\varepsilon}^{\infty}=\sqrt{1-1/\bar J_{\varepsilon}}$ and $f_0\le f_{\max,\delta}^{\infty}=\sqrt{1-2/\bar J_{\delta}}$, respectively. Thus the asymptotic feasible set is the interval $0\le f_0\le\min\{f_{\max,\varepsilon}^{\infty},f_{\max,\delta}^{\infty}\}$.

It remains to minimize the unconstrained asymptotic holding-cost proxy
\[
\psi(f_0)
:=
J_{D,\infty}(f_0;\beta)
=
\frac{(1+\beta-f_0)^2+\beta^2}{1-f_0^2}.
\]
Differentiating gives
\[
\psi'(f_0)
=
-\frac{
2\left[
(1+\beta)f_0^2
-2(\beta^2+\beta+1)f_0
+(1+\beta)
\right]
}{(1-f_0^2)^2}.
\]
For $\beta>0$, the unique critical point in $(0,1)$ is
\[
f_0^*(\beta)
=
\frac{\beta^2+\beta+1-\beta\sqrt{\beta^2+2\beta+2}}{1+\beta}.
\]
The derivative is negative on $(0,f_0^*)$ and positive on $(f_0^*,1)$, so $\psi(f_0)$ decreases up to $f_0^*$ and increases thereafter. Therefore, minimizing $\psi$ over the feasible interval gives $f_0^{\mathrm{opt},\infty}=\min\{f_0^*(\beta),f_{\max,\varepsilon}^{\infty},f_{\max,\delta}^{\infty}\}$.

{
\subsection{Proof of Proposition~\ref{prop:open-two-sided}}\label{app:open-two-sided}

Given $\mathbf{A}_{\mathrm{o}}^{\pm}=\left[\begin{smallmatrix}
f_0&f_{-1}&&0\\
f_1&f_0&\ddots&\\
&\ddots&\ddots&f_{-1}\\
0&&f_1&f_0
\end{smallmatrix}\right]_{N\times N}$. For its leading $j\times j$ block, define $D_j(\lambda)=\left|\begin{smallmatrix}
\lambda-f_0&-f_{-1}&&0\\
-f_1&\lambda-f_0&\ddots&\\
&\ddots&\ddots&-f_{-1}\\
0&&-f_1&\lambda-f_0
\end{smallmatrix}\right|_{j\times j}$. Cofactor expansion along the last row and then the last column of the off-diagonal minor gives $D_0(\lambda)=1$, $D_1(\lambda)=\lambda-f_0$, and $D_j(\lambda)=(\lambda-f_0)D_{j-1}(\lambda)-cD_{j-2}(\lambda)$ for $j=2,\ldots,N$, where $c=f_1f_{-1}$. If $c=0$, this recurrence gives $D_N(\lambda)=(\lambda-f_0)^N$, so all eigenvalues equal $f_0$.

If $c\ne0$, take $\xi=\sqrt{c}$, with $\xi=\mathrm{i}\sqrt{|c|}$ for $c<0$. The recurrence yields
\[
D_j(\lambda)=\xi^j U_j\!\left(\frac{\lambda-f_0}{2\xi}\right),
\]
where $U_j$ is the Chebyshev polynomial of the second kind, satisfying $U_0(x)=1$, $U_1(x)=2x$, and $U_j(x)=2xU_{j-1}(x)-U_{j-2}(x)$. The $N$ roots of $U_N$ are $\cos(k\pi/(N+1))$, $k=1,\ldots,N$, which proves \eqref{eq:open-two-sided-spectrum}.

For $c>0$, the eigenvalues are real and their extreme values are $f_0\pm2\sqrt{c}\,\chi_N$. Hence their largest modulus is $|f_0|+2\sqrt{c}\,\chi_N$. For $c<0$, each eigenvalue has real part $f_0$ and squared modulus $f_0^2+4|c|\cos^2(k\pi/(N+1))$, whose maximum is $f_0^2+4|c|\chi_N^2$. Requiring every eigenvalue to have modulus strictly below one, together with the $c=0$ case, gives \eqref{eq:open-two-sided-stability}.
}

\section{Loop Analysis: Proofs and Supplementary Details}\label{app:loop-analysis}

Headway and holding output formulas in this appendix concern the same-lap outputs $\bar{\boldsymbol\delta}_s$ and $\bar{\mathbf D}_s$ defined in Section~\ref{subsec:loop-performance-metrics}.

\subsection{Proof of Proposition~\ref{prop:loop-modal-decomposition}}\label{app:loop-modal-decomposition}

The matrix $\mathbf{P}$ shifts the entries of any vector one position downward in the bus-index ordering, wrapping the last entry to the top. Applying this shift to $\mathbf{u}_q$ gives $(\mathbf{P}\mathbf{u}_q)_n=(\mathbf{u}_q)_{n-1}$, with the index understood periodically. Hence $(\mathbf{P}\mathbf{u}_q)_n=e^{-\mathrm{i}\theta_q}(\mathbf{u}_q)_n=z_q(\mathbf{u}_q)_n$, and each Fourier vector $\mathbf{u}_q$ is an eigenvector of $\mathbf{P}$ with eigenvalue $z_q$. As the Fourier vectors form an orthonormal basis, $\mathbf{U}$ is unitary and diagonalizes $\mathbf{P}$: $\mathbf{U}^{*}\mathbf{P}\mathbf{U}=\diag(z_0,z_1,\ldots,z_{N-1})$.
Since $\mathbf{A}_{\ell}=F(\mathbf{P})$ is a polynomial in $\mathbf{P}$, the same basis diagonalizes $\mathbf{A}_{\ell}$:
\[
\mathbf{U}^{*}\mathbf{A}_{\ell}\mathbf{U}
=\sum_{p=0}^{m}f_p(\mathbf{U}^{*}\mathbf{P}\mathbf{U})^p
=\diag\bigl(F(z_q)\bigr)_{q=0}^{N-1}.
\]
Multiplying the state equation by $\mathbf{U}^{*}$ gives \eqref{eq:modal-dynamics}.

\subsection{Proof of Proposition~\ref{prop:loop-performance-metrics}}\label{app:loop-performance-metrics}

From \eqref{eq:modal-dynamics}, each mode satisfies
\[
\Var(\widehat{\varepsilon}_{q,s})
=
\sigma^2\sum_{j=0}^{s-1}|\lambda_q|^{2j}
=
\sigma^2\frac{1-|\lambda_q|^{2s}}{1-|\lambda_q|^2},
\]
with the continuity interpretation \(\sigma^2 s\) when \(|\lambda_q|=1\). Because \(\mathbf{U}\) is unitary, the trace of the covariance is unchanged by the DFT, which gives \eqref{eq:loop-schedule-finite}. If \eqref{eq:loop-stability} holds, then taking \(s\to\infty\) gives \eqref{eq:loop-schedule-ss}.

The headway vector is $\bar{\boldsymbol\delta}_s=\mathbf G_\ell\boldsymbol\varepsilon_s$. Its Fourier coordinates satisfy $\widehat{\bar\delta}_{q,s}=(1-z_q)\widehat\varepsilon_{q,s}$, giving \eqref{eq:loop-headway-finite} and, under strict stability, \eqref{eq:loop-headway-ss}.

The same-lap holding output satisfies $\bar{\mathbf D}_s-d_s\mathbf1=\mathbf M_\ell\boldsymbol\varepsilon_s$, with $\mathbf M_\ell$ defined in Section~\ref{subsec:loop-performance-metrics}. Its $q$-th Fourier multiplier is $\mu_q$, giving \eqref{eq:loop-holding-finite} and \eqref{eq:loop-holding-ss}.

\subsection{Proof of Proposition~\ref{prop:headway-based-zero-mode}}\label{app:headway-based-zero-mode}

Equating the two expressions $\bar{\mathbf D}_s-d_s\mathbf1=\mathbf M_\ell\boldsymbol\varepsilon_s$ and $\bar{\mathbf D}_s-d_s\mathbf1=K(\mathbf P)\mathbf G_\ell\boldsymbol\varepsilon_s$ gives $\mathbf A_\ell=(1+\beta)\mathbf I-\beta\mathbf P+K(\mathbf P)\mathbf G_\ell$. On a Fourier vector satisfying $\mathbf P\mathbf u=z\mathbf u$, this matrix has multiplier $(1+\beta)-\beta z+K(z)(1-z)$, proving \eqref{eq:headway-based-symbol}.

For the zero-frequency component, \(F_{\mathrm h}(z_0)=F_{\mathrm h}(1)=1\). Since its update multiplier is one, the loop transition matrix has an eigenvalue on the unit circle and cannot be Schur stable.

\subsection{Proof of Corollary~\ref{cor:headway-based-schedule-divergence}}\label{app:headway-based-schedule-divergence}

Proposition~\ref{prop:headway-based-zero-mode} gives \(\lambda_0=1\). The \(q=0\) term in \eqref{eq:loop-schedule-finite} is therefore \(s/N\), which gives \eqref{eq:headway-based-schedule-growth} and the divergence of \(J_{\varepsilon,s}^{\ell}\). For the same-lap headway metric, the $q=0$ term is multiplied by $h_0=0$, so it drops out of \eqref{eq:loop-headway-finite}. If all nonzero modes are stable, the remaining terms converge to the stated headway limit.

\subsection{BHC with Following-Bus Information}\label{app:bhc}\label{app:headway-based-examples}

{
These two subsections apply the headway-based analysis of Section~\ref{subsubsec:loop-headway-based-control} to the BHC and TWHC kernels defined in Section~\ref{subsec:local-model}. We use the same-lap outputs of Section~\ref{subsec:loop-performance-metrics} and assume exact following-bus information. For compactness, we retain the signed-offset notation $z^{-1}$, which equals $z^{N-1}$ at the loop's Fourier nodes.

For BHC, $F_{\mathrm{BH}}(z)=\alpha z^{-1}+(1+\beta-\alpha)-\beta z$. Since $F_{\mathrm{BH}}(1)=1$, the common schedule mode is neutral and its variance grows under persistent disturbances (Corollary~\ref{cor:headway-based-schedule-divergence}). For each nonzero mode, the headway weight $h_q=|1-z_q|^2$ gives the exact squared modal gain
\[
|\lambda_q^{\mathrm{BH}}|^2
=1+h_q\bigl[(\alpha+\beta)^2-(\alpha-\beta)-\alpha\beta h_q\bigr].
\]
Because $\alpha\beta>0$, all nonzero modes are stable exactly when
\[
\alpha-\beta>(\alpha+\beta)^2-\alpha\beta h_{\min},
\qquad h_{\min}=4\sin^2(\pi/N).
\]
For arbitrarily large fleets, the condition becomes $\alpha\in[\alpha_-,\alpha_+]$, with $\alpha_\pm=(1-2\beta\pm\sqrt{1-8\beta})/2$ and $0<\beta\le1/8$. At $\beta=1/8$, the endpoints coincide at $\alpha=3/8$. Gains below $\alpha_-$ cannot offset demand amplification, whereas gains above $\alpha_+$ overcorrect; either destabilizes sufficiently low-frequency modes as the fleet grows. The stability interval narrows as demand intensity $\beta$ increases, making BHC more sensitive to gain tuning. This sensitivity may partly explain the relatively poor performance reported for BHC in the literature.

Here $\mu_q^{\mathrm{BH}}=\alpha(z_q^{-1}-1)$, so each mode's holding-variance contribution is $\alpha^2$ times its headway-variance contribution. Weak low-frequency damping near the interval endpoints therefore raises both variances, while increasing $\alpha$ also raises this holding-cost factor. This combination helps explain BHC's high holding costs in Section~\ref{subsec:num-frontiers}.
}

\subsection{TWHC with Following-Bus Information}\label{app:twhc}

{
For TWHC, $F_{\mathrm{TW}}(z)=\alpha z^{-1}+(1-2\alpha)+\alpha z$, so $\lambda_q^{\mathrm{TW}}=1-\alpha h_q$. The common mode again has $\lambda_0=1$. For $0<\alpha<1/2$, every nonzero mode satisfies $0<h_q\le4$ and hence $|\lambda_q^{\mathrm{TW}}|<1$. Thus headway variance converges while schedule variance diverges. This characterization is consistent with the Fourier analysis of a related two-way-looking speed-control law by \citet{daganzo2011}.

The adjacent-contrast modes near $\theta_q=\pi$ have the largest headway weights and receive the largest corrections $\alpha h_q$. For $0<\alpha\le1/4$, these modes are also damped most strongly, aligning control with the headway metric. However, damping is strongest at $\alpha h_q=1$; beyond this point, the multiplier becomes negative and its magnitude increases. Increasing the feedback gain $\alpha$ therefore does not always improve headway regularity.

The holding-output multiplier is $\mu_q^{\mathrm{TW}}=\alpha z_q^{-1}-(\beta+2\alpha)+(\beta+\alpha)z_q$, giving
\[
|\mu_q^{\mathrm{TW}}|^2
=h_q\bigl[\beta^2+\alpha(\alpha+\beta)h_q\bigr].
\]
Each nonzero mode therefore contributes
\[
\frac1N\frac{|\mu_q^{\mathrm{TW}}|^2}{1-|\lambda_q^{\mathrm{TW}}|^2}
=\frac{\beta^2+\alpha(\alpha+\beta)h_q}{N\alpha(2-\alpha h_q)}
\]
to the normalized fleet-average holding-time variance. At low frequencies this contribution approaches $\beta^2/(2N\alpha)$, which also gives its leading growth as $\alpha\downarrow0$ for any fixed nonzero mode with $\beta>0$. Small $\alpha$ weakens headway restoration but retains demand compensation ($\mu_q^{\mathrm{TW}}\to\beta(z_q-1)$), so accumulated disturbances can produce large holding variance.

For even $N$, the alternating mode has $h_q=4$ and contributes $(\beta+2\alpha)^2/[2N\alpha(1-2\alpha)]$, which diverges as $\alpha\uparrow1/2$. Here $\lambda_{N/2}^{\mathrm{TW}}=1-4\alpha\to-1$, so the alternating mode becomes weakly damped as $\alpha$ approaches its upper stability limit. See Section~\ref{subsec:num-frontiers} for a more detailed numerical comparison of BHC and TWHC.
}

\subsection{Derivation of the Optimal Loop Control and Its Limiting Cases}\label{app:loop-optimal-design}

We prove Proposition~\ref{prop:loop-optimal-design} and derive the special cases and endpoints in Sections~\ref{subsubsec:loop-schedule-adhering} and~\ref{subsubsec:loop-headway-equalizing}, using the notation defined in the proposition. Since $G_q^2=1+\beta(1+\beta)h_q$, we have $G_q\ge1$, with strict inequality for nonzero modes when $\beta>0$.

The modal values \(\lambda_q=F_f(z_q)\) give an equivalent parametrization of a full cyclic kernel. The inverse DFT is one-to-one, and real coefficients correspond exactly to \(\lambda_{N-q}=\overline{\lambda_q}\), with \(\lambda_0\) real and, for even \(N\), \(\lambda_{N/2}\) real. For fixed \(|\lambda_q|=r\), the holding numerator \(|\lambda_q-g_q|^2=r^2+G_q^2-2\operatorname{Re}(\lambda_q\overline{g_q})\) is minimized by phase alignment. Thus \(\lambda_q=r_qg_q/G_q\), and the modal holding cost becomes \(\phi_{G_q}(r_q)=(r_q-G_q)^2/(1-r_q^2)\).

This reduction also establishes global optimality. Taking the second derivative shows that \(\phi_G(r)\) is strictly convex on \(0\le r<1\), and the reliability terms \(1/(1-r^2)\) and \(h_q/(1-r^2)\) are convex. Budgets strictly above the noise floors admit a strictly feasible point, satisfying Slater's condition. A finite schedule budget bounds all amplitudes away from one, so a unique optimum exists and the Karush--Kuhn--Tucker (KKT) conditions are necessary and sufficient. Symmetry of \(G_q\) and \(h_q\) (\(G_{N-q}=G_q, h_{N-q}=h_q\)) makes the minimizing amplitudes identical in conjugate mode pairs, preserving the real-kernel constraint.

Up to a positive scaling and constants independent of the modes, the Lagrangian separates as
\[
\sum_{q=0}^{N-1}
\frac{|\lambda_q-g_q|^2+W_q}{1-|\lambda_q|^2}.
\]
After phase alignment, each mode minimizes $[(r-G_q)^2+W_q]/(1-r^2)$ over $0\le r<1$. Its first-order condition is $G_qr^2-(1+G_q^2+W_q)r+G_q=0$, whose admissible root is \eqref{eq:modal-rstar}. This root lies strictly inside the unit disk when $G_q>1$ or $W_q>0$, and phase alignment gives \eqref{eq:modal-lambdastar}. Since $g_{N-q}=\overline{g_q}$ and $h_{N-q}=h_q$, the optimal modal values satisfy conjugate symmetry; the inverse DFT therefore gives the real coefficients in \eqref{eq:modal-fstar}.

For the zero mode, $G_0=1$, $h_0=0$, and the unpenalized holding cost is $(1-r_0)/(1+r_0)$, which decreases strictly with $r_0$. Thus a finite schedule budget requires $\eta_\varepsilon>0$, and $0<\lambda_0^*<1$. If the schedule constraint is omitted while full Schur stability is retained, this mode has only an infimum as $r_0\uparrow1$, not an attained minimum.

With only a headway constraint, the zero mode contributes no headway variance because $h_0=0$. For $-1<\lambda_0<1$, its holding-variance contribution is $(1-\lambda_0)/[N(1+\lambda_0)]$, strictly decreasing in $\lambda_0$. At $\lambda_0=1$, the holding output vanishes because $\mu_0=0$, attaining zero cost. Since a full cyclic kernel allows $\lambda_0$ to vary independently of the other modes, the unique optimal common-mode factor is $\lambda_0^*=1$. We therefore optimize the remaining modes $q=1,\ldots,N-1$ with $\eta_\varepsilon=0$. A finite headway-variance budget bounds those amplitudes away from one because \(h_q>0\). Strict convexity and Slater feasibility for \(R_{\delta}>2\) again give a unique optimum. The same phase and amplitude formulas apply to these modes, and \eqref{eq:modal-fstar} includes this optimal zero mode. To verify the frequency ordering of the damping, write the amplitude equation as \(r_q+r_q^{-1}=[2+(\beta(1+\beta)+\eta_{\delta})h_q]/\sqrt{1+\beta(1+\beta)h_q}\). Its right-hand side increases with \(h_q\), whereas \(r+r^{-1}\) decreases on \(0<r<1\); hence larger headway weights give smaller optimal amplitudes.

The tight endpoints follow directly from \eqref{eq:modal-rstar}. For schedule-adhering design, \(\eta_{\delta}=0\) and \(\eta_{\varepsilon}\to\infty\) make every amplitude vanish. Thus \(F_f^*\equiv0\), every \(f_p^*\to0\), and the noise-floor metrics are \(J_{\varepsilon,\infty}^{\ell}\to1\) and \(J_{\delta,\infty}^{\ell}\to2\). The holding cost tends to \(N^{-1}\sum_qG_q^2=1+2\beta+2\beta^2\). For headway-equalizing design, \(\eta_{\delta}\to\infty\) makes only the nonzero modes vanish, while \(\lambda_0^*=1\). The inverse transform gives \(f_p^*=1/N\), the fleet-mean projection; \(J_{\delta,\infty}^{\ell}\to2\) and \(J_{D,\infty}^{\ell}\to1+2\beta+2\beta^2-1/N\), while the absolute schedule variance remains unbounded. The optimal controls at the noise floors are attainable; in \eqref{eq:modal-rstar}, these boundary solutions are obtained by letting the corresponding multipliers tend to infinity. The same limiting interpretation applies when either budget reaches its floor in the two-constraint problem.

At either loose endpoint, the nonzero modes have \(W_q\to0\), so \(r_q^*\to1/G_q\) and
\[
\lambda_q^*\to\frac{g_q}{|g_q|^2}
=\frac{1}{(1+\beta)-\beta z_q^{-1}}.
\]
For schedule-adhering design, \(\lambda_0^*\to1\) only as the finite schedule budget is relaxed without bound; for headway-equalizing design it equals one throughout. The resulting inverse filter has the noncyclic expansion
\[
F^*(z)=\frac{1}{(1+\beta)-\beta z^{-1}}
=\frac{1}{1+\beta}\sum_{j=0}^{\infty}
\left(\frac{\beta}{1+\beta}\right)^jz^{-j}.
\]
The first unfolded weights are \(1/(1+\beta)\), \(\beta/(1+\beta)^2\), and \(\beta^2/(1+\beta)^3\). On a finite loop, terms with congruent offsets must be summed. With \(a=\beta/(1+\beta)<1\), the exact cyclic coefficients are
\[
f_{(-j)\bmod N}^*
=\frac{a^j}{(1+\beta)(1-a^N)},
\qquad j=0,\ldots,N-1.
\]
Thus the negative-offset terms encode geometrically weighted following-bus information; simply relabeling the unfolded weights modulo \(N\) would omit the repeated cycles.

Finally, at this inverse-filter endpoint each nonzero mode contributes \(G_q^2-1=\beta(1+\beta)h_q\) to the holding-cost sum and \(G_q^2/[\beta(1+\beta)]\) to the headway sum. Since \(\sum_qh_q=2N\), the endpoint metrics in Table~\ref{tab:loop-endpoints} are \(J_{D,\infty}^{\ell}=2\beta(1+\beta)\) and \(J_{\delta,\infty}^{\ell}=2+(N-1)/[N\beta(1+\beta)]\). For headway-variance budgets at or above this latter value, the inverse filter is already feasible and optimal with \(\eta_{\delta}=0\); tighter budgets require a positive multiplier. These calculations establish both special cases and all four endpoints without a separate optimization argument.

{
\section{Cross-lap Feedback and Performance Corrections}\label{app:physical-loop}

Appendix~\ref{app:seam-delay} examines the stability of cross-lap feedback laws. Appendices~\ref{app:physical-metrics}--\ref{app:physical-bounds} quantify how cross-lap corrections to headway and holding outputs affect performance and optimal design, retaining same-lap kernel inputs. Throughout, $N\ge2$, $S\ge1$ is the number of control points per circuit, $\beta>0$, and the disturbance assumptions of Proposition~\ref{prop:loop-performance-metrics} apply.

\subsection{Cross-Lap Feedback}\label{app:seam-delay}

We model a cross-lap feedback law using observed preceding-lap inputs and test a two-tap kernel on whether the open-chain/loop stability contrast survives this change. The cross-lap law replaces each wrapping leader input by its observed previous-lap value. For example, the $f_1$ term at bus~1 uses $\varepsilon_{N,s-S}$ instead of $\varepsilon_{N,s}$; the two-tap laws coincide at all other buses. The demand-compensation term uses the observed leader under either feedback law. For physical leading-bus offsets $0\le m\le N-1$, this substitution gives

\begingroup
\begin{equation}
\boldsymbol\varepsilon_{s+1}=(\mathbf A_\ell-\mathbf W)\boldsymbol\varepsilon_s
+\mathbf W\boldsymbol\varepsilon_{s-S}+\mathbf v_{s+1},\quad
\mathbf W=\sum_{p=1}^{m}f_p(\mathbf P^p-\mathbf L^p).
\label{eq:seam-delay}
\end{equation}
\endgroup
Here $\mathbf P^p-\mathbf L^p$ retains only the $p$-step bus-index shifts that cross the seam; these inputs are evaluated at $s-S$ rather than $s$ in \eqref{eq:seam-delay}. Thus $\mathbf W$ is supported within the first $m$ rows and last $m$ columns. This construction uses physical leader offsets before modulo rewriting.

To retain the history needed at successive steps, stack the current state and the preceding $S$ states into $(\boldsymbol\varepsilon_s^\top,\ldots,\boldsymbol\varepsilon_{s-S}^\top)^\top$. This augmented state contains $S+1$ blocks of $N$ components each, so its dimension is $N(S+1)$. The resulting first-order matrix recursion admits numerical spectral and covariance calculations, but for general kernels the seam-specific delay removes the circulant structure underlying the closed-form performance and design analysis. The main text therefore focuses on same-lap feedback.

For the forward-looking two-tap kernel $F(z)=f_0+f_1z$, $f_1\ne0$, one has $\mathbf A_\ell-\mathbf W=f_0\mathbf I+f_1\mathbf L$ and $\mathbf W=f_1\mathbf e_1\mathbf e_N^\top$, where $\mathbf e_n$ has a one in position $n$ and zeros elsewhere.

As in Proposition~\ref{prop:loop-stability}, stability is determined by the eigenvalues of the disturbance-free recursion. Seek a mode $\boldsymbol\varepsilon_s=\zeta^s\boldsymbol\phi$, where $\zeta\in\mathbb C\setminus\{0\}$ is its one-step multiplier and $\boldsymbol\phi\in\mathbb C^N\setminus\{\mathbf0\}$ specifies the relative amplitudes and phases across buses. Every component of this mode is multiplied by $\zeta$ at each step; hence it decays if $|\zeta|<1$ and grows if $|\zeta|>1$. Substitution gives $(\zeta-f_0)\phi_n=f_1\phi_{n-1}$ for $n\ge2$ and $(\zeta-f_0)\phi_1=f_1\zeta^{-S}\phi_N$. Iterating the first relation yields $(\zeta-f_0)^{N-1}\phi_N=f_1^{N-1}\phi_1$; combining it with the seam relation gives the equation for the nonzero augmented-state eigenvalues:
\begin{equation}
\zeta^S(\zeta-f_0)^N=f_1^N.
\label{eq:app-seam-char}
\end{equation}
Under the conditions below, a kernel stable on the open chain remains unstable on the loop under either feedback convention. Suppose $0<f_0<1$, $|f_1|>1-f_0$, and $f_1^N>0$. The open-chain spectral radius is $f_0<1$, whereas the same-lap loop spectral radius is $f_0+|f_1|>1$. For cross-lap feedback, $x^S(x-f_0)^N$ is strictly increasing for $x>1$, starting below $f_1^N$ at $x=1$ and tending to infinity; thus \eqref{eq:app-seam-char} has a unique real root $x^*>1$. For any root with $|\zeta|>1$, the reverse triangle inequality gives $|\zeta|^S(|\zeta|-f_0)^N\le|f_1|^N$, so $|\zeta|\le x^*$. Thus $x^*$ is the spectral radius of the augmented-state transition matrix and the one-step growth factor of its fastest-growing mode. With the kernel and $N$ fixed, $x^*$ decreases with $S$, and $1<x^*<f_0+|f_1|$ for $S>0$. Thus, for this parameter class, \textbf{cross-lap feedback slows the loop instability but does not remove the open-chain/loop contrast}.

For the numerical case in Section~\ref{subsec:num-boundary} with $K_2=(0.25,-0.8)$ and $N=S=12$, the open chain is stable with spectral radius $0.25$, the same-lap loop is unstable with spectral radius $1.05$, and the loop under cross-lap feedback remains unstable with $x^*\approx1.028$. Under the stated white-noise forcing, the two unstable loops have asymptotic variance growth factors of approximately $1.103$ and $1.057$ per passage, respectively. This example confirms that the loop instability is not an artifact of the same-lap feedback assumption.

For the three-tap examples in Section~\ref{subsec:loop-stability-xuan}, physical neighboring-bus inputs require both $\varepsilon_{N,s-S}$ in bus~1's leader term and $\varepsilon_{1,s+S}$ in bus~$N$'s follower term, with the latter predicted exactly. For $N=S=3$, reindex the state as $\mathbf y_s=(\varepsilon_{1,s+2},\varepsilon_{2,s+1},\varepsilon_{3,s})^\top$. The disturbance-free equations become
\[
(\mathbf I-f_{-1}\mathbf P^\top)\mathbf y_{s+1}
=f_0\mathbf y_s+f_1\mathbf P\mathbf y_{s-1}.
\]
For either kernel, the matrix on the left is invertible, and the resulting six-dimensional state transition has eigenvalues given by the roots of
\[
(1-f_{-1}z_q^{-1})\zeta^2-f_0\zeta-f_1z_q=0,
\qquad q=0,1,2.
\]
The largest root moduli are approximately $1.102879$ for $(f_0,f_1,f_{-1})=(0.5,0.3,0.3)$ and $0.944996$ for $(0.4,-0.5,-0.5)$, giving the two cross-lap spectral radii reported in the main text.

\subsection{Headway and Holding Outputs with Cross-Lap Corrections}\label{app:physical-metrics}

Appendices~\ref{app:physical-metrics}--\ref{app:physical-bounds} retain same-lap feedback and assume exact predictions of all required but unobserved same-lap inputs, including those at the cyclic boundary. We correct only the headway output and demand-compensation term using the actual immediate leader, leaving the kernel inputs unchanged; cross-lap feedback is treated separately in Appendix~\ref{app:seam-delay}.

We use the superscript $\mathrm{xc}$ for outputs and metrics with cross-lap corrections. At bus~1, the demand-compensation term of the holding law uses the preceding-lap leader deviation $\varepsilon_{N,s-S}$, giving
\[
D_{1,s}^{\mathrm{xc}}-d_s=-(1+\beta)\varepsilon_{1,s}
+\beta\varepsilon_{N,s-S}+\mathbf e_1^\top\mathbf A_\ell\boldsymbol\varepsilon_s.
\]
The compensation cancels demand at every bus, preserving the same-lap state dynamics \eqref{eq:loop-dynamics}. Let $\omega_s=\varepsilon_{N,s-S}-\varepsilon_{N,s}$. Using the same-lap output matrices $\mathbf G_\ell$ and $\mathbf M_\ell$ from Section~\ref{subsec:loop-performance-metrics}, the outputs with cross-lap corrections are
\begingroup
\begin{equation}
\boldsymbol\delta_s^{\mathrm{xc}}=\mathbf G_\ell\boldsymbol\varepsilon_s-\mathbf e_1\omega_s,\qquad
\mathbf D_s^{\mathrm{xc}}-d_s\mathbf1=\mathbf M_\ell\boldsymbol\varepsilon_s+\beta\mathbf e_1\omega_s.
\label{eq:physical-outputs}
\end{equation}
\endgroup
Only the first component of each output vector differs from the corresponding same-lap output, regardless of kernel support; its covariance with other deviations depends on the full kernel. Write $J_{\delta,s}^{\mathrm{xc}}$ and $J_{D,s}^{\mathrm{xc}}$ for the fleet-average variances of these outputs divided by $\sigma^2$, and retain $J_{\delta,s}^{\ell}$ and $J_{D,s}^{\ell}$ for the same-lap metrics.

The output identities \eqref{eq:physical-outputs} apply without a stability assumption. For $s\ge S$, finite-horizon variances with cross-lap corrections follow from them using $\boldsymbol\Sigma_s=\Var(\boldsymbol\varepsilon_s)$, $\boldsymbol\Sigma_{s-S}$, and $\operatorname{Cov}(\boldsymbol\varepsilon_s,\boldsymbol\varepsilon_{s-S})=\mathbf A_\ell^S\boldsymbol\Sigma_{s-S}$. During the first circuit ($0\le s<S$), we set the required preceding-circuit deviations $\varepsilon_{N,s-S}$ to zero.

\begin{proposition}[Exact cross-lap variance corrections]\label{prop:physical-corrections}
Assume strictly stable dynamics, $\max_q|\lambda_q|<1$, and put $v_q=(1-|\lambda_q|^2)^{-1}$ and $\Delta J_X=J_{X,\infty}^{\mathrm{xc}}-J_{X,\infty}^{\ell}$, $X\in\{\delta,D\}$. Then
\begin{align}
\Delta J_\delta&=\frac2{N^2}\sum_qv_q\operatorname{Re}\{z_q^{-1}(1-\lambda_q^S)\},\notag\\
\Delta J_D&=\frac{2\beta}{N^2}\sum_qv_q\operatorname{Re}\{z_q^{-1}(1+\beta-\lambda_q)(1-\lambda_q^S)\}.
\label{eq:physical-corrections}
\end{align}
Both corrections vanish for SC with $F(z)=f_0$, $|f_0|<1$.
\end{proposition}

\begin{proof}
Let $\boldsymbol\Sigma_\infty=\lim_{s\to\infty}\Var(\boldsymbol\varepsilon_s)$ be the stationary covariance of the loop state. In steady state, $\operatorname{Cov}(\boldsymbol\varepsilon_s,\boldsymbol\varepsilon_{s-S})=\mathbf A_\ell^S\boldsymbol\Sigma_\infty$. Current and preceding-lap deviations of bus~$N$ have equal variances, so their substitution changes only the cross-covariance terms:
\[
\begin{aligned}
\Delta J_\delta&=\frac2{N\sigma^2}\mathbf e_1^\top(\mathbf I-\mathbf A_\ell^S)\boldsymbol\Sigma_\infty\mathbf e_N,\\
\Delta J_D&=\frac{2\beta}{N\sigma^2}\mathbf e_1^\top[(1+\beta)\mathbf I-\mathbf A_\ell]
(\mathbf I-\mathbf A_\ell^S)\boldsymbol\Sigma_\infty\mathbf e_N.
\end{aligned}
\]
Fourier diagonalization gives \eqref{eq:physical-corrections}; under SC, $\lambda_q=f_0$ and $\sum_qz_q^{-1}=0$, so $\Delta J_\delta=\Delta J_D=0$.
\end{proof}

The corrections can have either sign. Their $\lambda_q^S$ terms preserve exact modal evaluation but prevent direct reuse of the phase-alignment and quadratic-root formulas derived in Appendix~\ref{app:loop-optimal-design} for the optimum in Section~\ref{subsec:loop-modal-optimization}.

\subsection{Cross-Lap Corrections with a Neutral Common Mode \texorpdfstring{($\lambda_0=1$)}{(lambda0=1)}}\label{app:physical-neutral}

Proposition~\ref{prop:physical-corrections} in Appendix~\ref{app:physical-metrics} assumes that every mode is strictly stable. We now consider $\lambda_0=1$ with all nonzero modes stable, as in headway-based control. The schedule-deviation variance then diverges under persistent disturbances, but headway and holding variances still converge with cross-lap corrections.

Write $\bar\varepsilon_s=N^{-1}\sum_{n=1}^N\varepsilon_{n,s}$ for the fleet-average schedule deviation. Its increments are independent, each with variance $\sigma^2/N$. For $s\ge S$, the difference $\bar\varepsilon_s-\bar\varepsilon_{s-S}$ contains only the latest $S$ increments and has variance $S\sigma^2/N$, independent of the number of circuits completed. Same-lap adjacent differences cancel this common offset. At the seam, however, its contributions to bus~$1$'s headway and holding outputs with cross-lap corrections are $\bar\varepsilon_s-\bar\varepsilon_{s-S}$ and $\beta(\bar\varepsilon_{s-S}-\bar\varepsilon_s)$, respectively. For this mode, all other buses have zero output contributions. Averaging the output variances over $N$ buses and dividing by $\sigma^2$ therefore gives
\[
\frac{S}{N^2},\qquad\frac{\beta^2S}{N^2},
\]
respectively. These are fixed variance contributions, not amounts accumulated each circuit.

The limiting variances with cross-lap corrections are obtained by adding these constants to the corresponding nonzero-mode contributions: retain only $q=1,\ldots,N-1$ in both the formulas of Proposition~\ref{prop:loop-performance-metrics} and the corrections in \eqref{eq:physical-corrections}. Under the stated disturbance assumptions, different modes are uncorrelated, so their variance contributions add.

Within the class $\lambda_0=1$, the holding objective is shifted by $\beta^2S/N^2$ without changing its minimizer, while the headway-variance budget available to relative modes is reduced by $S/N^2$. The nonzero-mode corrections depend on the control coefficients and must still enter optimization with cross-lap corrections.

\subsection{Bounds and Design Guarantees after Cross-Lap Corrections}\label{app:physical-bounds}\label{app:physical-design}

Proposition~\ref{prop:physical-uniform-bounds} bounds the performance errors due to cross-lap corrections; Proposition~\ref{prop:physical-design-transfer} then guarantees feasibility and bounds the excess holding cost of Proposition~\ref{prop:loop-optimal-design}'s closed-form design, computed with a tightened headway-variance budget, relative to the optimum with cross-lap corrections.

\begin{proposition}[Bounds under a schedule budget]\label{prop:physical-uniform-bounds}
For any strictly stable cyclic kernel satisfying $J_{\varepsilon,\infty}^{\ell}\le R<\infty$, the following bounds hold uniformly in $S$:
\begin{equation}
\begin{aligned}
|\Delta J_\delta|&\le\epsilon_\delta:=\frac{2}{N}\bigl[R-1+\sqrt{R(R-1)}\bigr],\\
|\Delta J_D|&\le\epsilon_D:=\beta(2+\beta)\epsilon_\delta.
\end{aligned}
\label{eq:physical-uniform-bounds}
\end{equation}
\end{proposition}

\begingroup
\begin{proof}
Set $\mathbf C=\boldsymbol\Sigma_\infty/\sigma^2$ and $v_q=(1-|\lambda_q|^2)^{-1}$. The Fourier decomposition gives $\mathbf C=\mathbf U\diag(v_q)\mathbf U^*$, whose diagonal entries all equal $\bar v:=N^{-1}\sum_qv_q=J_{\varepsilon,\infty}^{\ell}\le R$. Since $v_q\ge1$, the matrix $\mathbf C-\mathbf I$ is positive semidefinite. For $n\ne n'$, the positive-semidefinite covariance bound gives $|C_{nn'}|^2\le(C_{nn}-1)(C_{n'n'}-1)=(\bar v-1)^2$, hence $|C_{nn'}|\le R-1$.

Next, $\mathbf A_\ell^k\mathbf C$ is the normalized lag-$k$ covariance matrix. Using $|U_{nq}|=1/\sqrt N$, its entries satisfy the triangle-inequality bound
\[
| (\mathbf A_\ell^k\mathbf C)_{nn'} |
=\left|\sum_q U_{nq}\lambda_q^kv_q\overline{U_{n'q}}\right|
\le\frac1N\sum_q|\lambda_q|^kv_q.
\]
For $k\ge1$, stability gives $|\lambda_q|^k\le|\lambda_q|$, while $|\lambda_q|v_q=\sqrt{v_q(v_q-1)}$. Cauchy--Schwarz therefore yields
\[
\frac1N\sum_q|\lambda_q|^kv_q
\le\frac1N\sum_q\sqrt{v_q(v_q-1)}
\le\sqrt{\left(\frac1N\sum_qv_q\right)\left(\frac1N\sum_q(v_q-1)\right)}
\le\sqrt{R(R-1)}.
\]
For $k\ge2$, the sharper inequality $|\lambda_q|^k\le|\lambda_q|^2$ gives $N^{-1}\sum_q|\lambda_q|^kv_q\le N^{-1}\sum_q(v_q-1)=\bar v-1\le R-1$.

Apply these bounds to the $(1,N)$ entries in the correction formulas of Appendix~\ref{app:physical-metrics}. Write $a=R-1$ and $b=\sqrt{R(R-1)}$. Since $S\ge1$, the headway correction satisfies $|\Delta J_\delta|\le2(a+b)/N=\epsilon_\delta$. For holding, expand
\[
[(1+\beta)\mathbf I-\mathbf A_\ell](\mathbf I-\mathbf A_\ell^S)\mathbf C
=(1+\beta)\mathbf C-\mathbf A_\ell\mathbf C
-(1+\beta)\mathbf A_\ell^S\mathbf C+\mathbf A_\ell^{S+1}\mathbf C.
\]
The four $(1,N)$ entries on the right are bounded in modulus by $(1+\beta)a$, $b$, $(1+\beta)b$, and $a$, respectively; the last uses $S+1\ge2$. Thus
\[
|\Delta J_D|\le\frac{2\beta}{N}\bigl[(1+\beta)a+b+(1+\beta)b+a\bigr]
=\beta(2+\beta)\epsilon_\delta.
\]
These bounds do not depend on $S$ or the individual kernel coefficients.
\end{proof}
\endgroup

Thus a common schedule budget bounds the fleet-average discrepancy by $O(N^{-1})$, independently of kernel support. Note that this is not a uniform claim without such a budget.

Let $\mathcal D$ be a class of cyclic kernels on which the uniform error bounds $|J_{\delta,\infty}^{\mathrm{xc}}-J_{\delta,\infty}^{\ell}|\le\epsilon_\delta$ and $|J_{D,\infty}^{\mathrm{xc}}-J_{D,\infty}^{\ell}|\le\epsilon_D$ hold. The class $\mathcal D$ contains the reference optimizer $\widehat F$ defined below and the comparison kernels $F$, which satisfy the original budgets after cross-lap corrections.

\begin{proposition}[Feasibility and an additive design guarantee]\label{prop:physical-design-transfer}
Let $\widehat F\in\mathcal D$ solve the reference problem with $J_{\varepsilon,\infty}^{\ell}\le R_\varepsilon$ and $J_{\delta,\infty}^{\ell}\le R_\delta-\epsilon_\delta$, omitting absent constraints. Assume feasibility and nonnegative multipliers satisfying complementary slackness at $\widehat F$ and the global inequality
\begingroup
\[
J_{D,\infty}^{\ell}(F)+\eta_\varepsilon[J_{\varepsilon,\infty}^{\ell}(F)-R_\varepsilon]
+\eta_\delta[J_{\delta,\infty}^{\ell}(F)-R_\delta+\epsilon_\delta]
\ge J_{D,\infty}^{\ell}(\widehat F),\quad F\in\mathcal D.
\]
\endgroup
Then $\widehat F$ satisfies the original budgets after accounting for cross-lap corrections. Let $V^{\mathrm{xc}}$ be the infimum of $J_{D,\infty}^{\mathrm{xc}}$ over kernels in $\mathcal D$ satisfying those original budgets. Then
\begin{equation}
0\le J_{D,\infty}^{\mathrm{xc}}(\widehat F)-V^{\mathrm{xc}}
\le2\epsilon_D+2\eta_\delta\epsilon_\delta.
\label{eq:physical-design-gap}
\end{equation}
Schedule-only design has $\eta_\delta=0$, requires no tightening, and has gap at most $2\epsilon_D$.
\end{proposition}

\begingroup
\begin{proof}
The tightened budget and the headway error bound ensure that $\widehat F$ satisfies all imposed original budgets.

Now take any $F\in\mathcal D$ feasible under those budgets with cross-lap corrections. Its reference headway metric satisfies $J_{\delta,\infty}^{\ell}(F)\le J_{\delta,\infty}^{\mathrm{xc}}(F)+\epsilon_\delta\le R_\delta+\epsilon_\delta$, so it exceeds the tightened budget by at most $2\epsilon_\delta$. When imposed, its schedule-constraint term is nonpositive. Using nonnegative multipliers in the assumed global inequality therefore gives $J_{D,\infty}^{\ell}(\widehat F)\le J_{D,\infty}^{\ell}(F)+2\eta_\delta\epsilon_\delta$. Applying the holding error bound once to each design yields
\[
J_{D,\infty}^{\mathrm{xc}}(\widehat F)
\le J_{D,\infty}^{\ell}(\widehat F)+\epsilon_D
\le J_{D,\infty}^{\mathrm{xc}}(F)+2\epsilon_D+2\eta_\delta\epsilon_\delta.
\]
Taking the infimum over all such feasible $F$ proves the upper bound in \eqref{eq:physical-design-gap}; feasibility of $\widehat F$ gives the lower bound. If the headway constraint is absent, set $\eta_\delta=0$ and omit its budget tightening.
\end{proof}

The global inequality required by Proposition~\ref{prop:physical-design-transfer} is guaranteed by Proposition~\ref{prop:loop-optimal-design} for the full cyclic design. Let $\mathcal L(F)$ denote its left-hand side with fixed multipliers. Under a finite schedule budget,
\[
\mathcal L(F)=\frac1N\sum_q\frac{|\lambda_q-g_q|^2+W_q}{1-|\lambda_q|^2}
-\eta_\varepsilon R_\varepsilon-\eta_\delta(R_\delta-\epsilon_\delta),
\qquad W_q=\eta_\varepsilon+\eta_\delta h_q.
\]
Equations~\eqref{eq:modal-rstar}--\eqref{eq:modal-lambdastar} minimize each summand (Appendix~\ref{app:loop-optimal-design}); complementary slackness then gives $\mathcal L(F)\ge\mathcal L(\widehat F)=J_{D,\infty}^{\ell}(\widehat F)$. For headway-only design, apply this argument to nonzero modes and minimize the common-mode holding term at $\lambda_0=1$.

Proposition~\ref{prop:physical-design-transfer} allows the closed-form design of Proposition~\ref{prop:loop-optimal-design} to be used with a tightened headway-variance budget, guaranteeing feasibility and the optimality-gap bound \eqref{eq:physical-design-gap} after cross-lap corrections. This avoids a more complex optimization that may lack a closed-form solution. Under a finite schedule budget, take $\mathcal D$ as all strictly stable cyclic kernels with $J_{\varepsilon,\infty}^{\ell}\le R_\varepsilon$. Obtain the uniform errors from Proposition~\ref{prop:physical-uniform-bounds} with $R=R_\varepsilon$ and use $R_\delta-\epsilon_\delta$ in Proposition~\ref{prop:loop-optimal-design}, provided its assumptions hold. The optimum with cross-lap corrections is attained because $v_q\le NR_\varepsilon$ makes the feasible set compact and the cost is continuous.
\endgroup

It remains to obtain error bounds for headway-only design, where schedule variance may diverge and Proposition~\ref{prop:physical-uniform-bounds} does not apply. Bounds on headway and holding variances suffice to extend the guarantee in Proposition~\ref{prop:physical-design-transfer} to this case. Consider kernels with stable nonzero modes and $\lambda_0\in(-1,1]$, including the neutral common mode in Appendix~\ref{app:physical-neutral}. If $J_{\delta,\infty}^{\ell}\le B_\delta$ and $J_{D,\infty}^{\ell}\le B_D$ throughout the comparison domain, valid uniform errors are
\begingroup
\begin{equation}
B_\omega:=S+S^2(\sqrt{B_D}+\beta\sqrt{B_\delta})^2,\quad
\epsilon_\delta:=\frac{B_\omega+2\sqrt{B_\delta B_\omega}}N,\quad
\epsilon_D:=\frac{\beta^2B_\omega+2\beta\sqrt{B_DB_\omega}}N.
\label{eq:physical-cost-sublevel-bounds}
\end{equation}

To derive these errors, first bound the variance of the seam difference $\omega_s$. The limit of $\Var(\omega_s)/\sigma^2$ is the modal average of
\[
\frac{2(1-\operatorname{Re}\lambda^S)}{1-|\lambda|^2}
=\sum_{k=0}^{S-1}|\lambda|^{2k}
+\frac{|1-\lambda^S|^2}{1-|\lambda|^2},
\]
with common-mode value $S$ when $\lambda_0=1$. The first term is at most $S$. For the second, $|1-\lambda^S|\le S|1-\lambda|$ and $1-\lambda_q=-\mu_q+\beta(z_q-1)$ give, by the weighted triangle inequality,
\[
\frac1N\sum_q\frac{|1-\lambda_q|^2}{1-|\lambda_q|^2}
\le\left(\sqrt{J_{D,\infty}^{\ell}}+\beta\sqrt{J_{\delta,\infty}^{\ell}}\right)^2
\le(\sqrt{B_D}+\beta\sqrt{B_\delta})^2,
\]
where the neutral common-mode term is zero by continuity. Hence $\lim_{s\to\infty}\Var(\omega_s)/\sigma^2\le B_\omega$. In \eqref{eq:physical-outputs}, the added $\omega_s$ term changes each variance through its own variance and twice its covariance with the reference output. Cauchy--Schwarz bounds the covariance using $B_\delta$ or $B_D$; division by $N$ then gives \eqref{eq:physical-cost-sublevel-bounds}, since only bus~$1$ changes.
\endgroup

These caps must also cover the designs needed to determine the optimum with cross-lap corrections. Choose a feasible design with holding cost $\Gamma$ under those corrections. Designs costing more than $\Gamma$ cannot improve on it, so restricting comparison to costs at most $\Gamma$ leaves the infimum unchanged. With $a=(N-1)/N$, the $N-1$ unchanged bus outputs imply $J_{\delta,\infty}^{\mathrm{xc}}\ge aJ_{\delta,\infty}^{\ell}$ and $J_{D,\infty}^{\mathrm{xc}}\ge aJ_{D,\infty}^{\ell}$. Thus every remaining feasible design satisfies $J_{\delta,\infty}^{\ell}\le R_\delta/a$ and $J_{D,\infty}^{\ell}\le\Gamma/a$.

Choose $B_\delta\ge R_\delta/a$ and $B_D\ge\Gamma/a$, and verify that the reference optimizer $\widehat F$ also satisfies these caps. Equation~\eqref{eq:physical-cost-sublevel-bounds} then supplies the uniform errors required by Proposition~\ref{prop:physical-design-transfer}. Provided the tightened reference problem satisfies Proposition~\ref{prop:loop-optimal-design}'s assumptions, its closed-form design remains feasible after cross-lap corrections and has the gap bound \eqref{eq:physical-design-gap}, even when schedule variance diverges. The bounds are conservative, so they need not yield a useful guarantee at every parameter setting. For $N=12$, $\beta=0.05$, $S=1$, and $R_\delta=5$, using caps $B_\delta=60/11$ and $B_D=0.25$ ensures the coverage of the tightened optimizer and all competitive designs: $\epsilon_\delta=0.5724$ and $\epsilon_D=0.005183$ yield $J_{\delta,\infty}^{\mathrm{xc}}=4.433$ and an additive gap bound of $0.05317$. With cross-lap corrections included, the tightened analytical design's holding cost exceeds that of the best-found full cyclic numerical design under the original budget by $0.01778$. Section~\ref{subsec:num-frontiers} evaluates less conservative retuning using exact corrections, with smaller observed cost gaps at moderate and loose budgets.

}

\section{Finite-Horizon Boundary Effects}\label{app:finite-horizon}

\subsection{Proof of Proposition~\ref{prop:transient-envelope}}\label{app:transient-envelope}

Under open-chain truncation, $(\mathbf A_{\mathrm o})_{ab}=f_{a-b}$, taking $f_p=0$ outside $-r,\ldots,m$; on the loop, $\mathbf A_\ell=\sum_{p=-r}^{m}f_p\mathbf P^p$. In a $j$-step path contributing to row $n$, each shift lies between $-r$ and $m$, so every intermediate source index lies between $n-mj$ and $n+rj$. If $mj\le n-1$ and $rj\le N-n$, all such paths stay within buses $1,\ldots,N$: truncation removes none and periodic indexing adds no wrap-around paths. Consequently, $\mathbf e_n^\top\mathbf A_{\mathrm o}^j=\mathbf e_n^\top\mathbf A_\ell^j$ for $j\le s_n^\dagger-1$, where $\mathbf e_n$ is the $n$th coordinate vector.

The row $\mathbf e_n^\top\mathbf A^j$ gives the weights on disturbances $j$ passages earlier. Under the initial-state and disturbance assumptions of Proposition~\ref{prop:transient-envelope}, either boundary therefore gives
\[
\Var(\varepsilon_{n,s})=\sigma^2\sum_{j=0}^{s-1}\bigl\lVert\mathbf e_n^\top\mathbf A^j\bigr\rVert_2^2,
\qquad \mathbf A\in\{\mathbf A_{\mathrm o},\mathbf A_\ell\}.
\]
For $s\le s_n^\dagger$, the row identity makes each summand equal under the two boundaries. Circulant symmetry makes the loop variance equal across buses, proving \eqref{eq:coincidence-window}.

Schur stability gives a finite open-chain variance limit, no smaller than the variance at $s_n^\dagger$ because the summands above are nonnegative. Using \eqref{eq:coincidence-window} and \eqref{eq:loop-schedule-finite},
\begingroup
{
\[
\lim_{s\to\infty}\Var(\varepsilon_{n,s})
\ge\Var(\varepsilon_{n,s_n^\dagger})
=\frac{\sigma^2}{N}\sum_{q=0}^{N-1}\sum_{j=0}^{s_n^\dagger-1}|F(z_q)|^{2j}
\ge\frac{\sigma^2}{N}\sum_{j=0}^{s_n^\dagger-1}M^{2j}
=\frac{\sigma^2}{N}\frac{M^{2s_n^\dagger}-1}{M^2-1}.
\]
}
\endgroup
The last inequality keeps a mode attaining $M$; $M>1$ gives the geometric sum. For a fixed kernel with $\sup_{|z|=1}|F(z)|>1$, continuity gives $M\ge1+\eta$ for some fixed $\eta>0$ on sufficiently fine Fourier grids. Hence the bound grows geometrically with $N$ when $s_n^\dagger$ grows proportionally with $N$ and the open chains remain stable.

\end{document}